\documentclass[journal,10pt,draftclsnofoot,onecolumn]{IEEEtran}
\usepackage[utf8]{inputenc}
\usepackage[T1]{fontenc}
\usepackage[english]{babel}

\usepackage{amsmath}
\usepackage{amssymb,amsthm,color}
\usepackage{bbold}
\usepackage{xspace}
\usepackage{xparse}
\usepackage{xr-hyper}
\usepackage[hidelinks]{hyperref}
\hypersetup{colorlinks=true,citecolor=blue,linkcolor=blue,filecolor=blue,urlcolor=blue}
\usepackage{url}
\usepackage{lmodern}
\usepackage{graphicx}
\usepackage[usenames,dvipsnames]{xcolor}
\usepackage{cleveref}
\usepackage{enumitem}
\usepackage[normalem]{ulem}
\usepackage{tablefootnote}
\usepackage{booktabs}

\definecolor{darkgreen}{rgb}{0,0.5,0}

\newcommand\textred[1]{{\color{red}\bfseries#1}}
\newcommand\mathred[1]{{\color{red}\boldsymbol{#1}}}
\newcommand\red[1]{\ifmmode\mathred{#1}\else\textred{#1}\fi}
\newcommand\mred[1]{\marginpar{\tiny{\ifmmode\mathred{#1}\else\textred{#1}\fi}}}
\newcommand\textblue[1]{{\color{blue}\bfseries#1}}
\newcommand\mathblue[1]{{\color{blue}\boldsymbol{#1}}}
\newcommand\blue[1]{\ifmmode\mathblue{#1}\else\textblue{#1}\fi}

\usepackage{stmaryrd}
\def\eps{\varepsilon} 
\def\P{{\mathbb{P}}} 
\def\E{{\mathbb{E}}} 
\def\R{{\mathbb R}} 
\def\F{{\mathcal F}} 
\def\Z{{\mathbb Z}} 
\def\N{{\mathbb N}} 
\def\D{\mathcal D}
\def\mQ{{\mathcal Q}}%
\renewcommand{\leq}{\leqslant}%
\renewcommand{\geq}{\geqslant}%
\newcommand{\ceil}[1]{{\lceil #1 \rceil}} 
\newcommand{\floor}[1]{{\lfloor #1 \rfloor}} 
\def\Lips{\text{Lips}} 
\newcommand{\vertiii}[1]{{\left\vert\kern-0.25ex\left\vert\kern-0.25ex\left\vert #1
    \right\vert\kern-0.25ex\right\vert\kern-0.25ex\right\vert}} 
\def\encodingspeed{\ensuremath{\gamma^{*\text{encod}}}} 
\def\approximationspeed{\ensuremath{\gamma^{*\text{approx}}}} 
\def\speed{\ensuremath{\gamma}} 
\def\fctclass{{\mathcal C}}
\def\approxclass{{\Sigma}}

\def\nnclass{\mathcal N}

\DeclareMathOperator*{\esssup}{ess\,sup}
\NewDocumentCommand{\param}{o o}{%
	\IfNoValueTF{#1}{%
		\theta\xspace
	}{%
		\IfNoValueTF{#2}{%
			\theta_{#1}\xspace
		}{%
			\theta_{#1}^{#2}\xspace
		}
	}
}
\NewDocumentCommand{\archiclass}{m}{%
	\mathtt{A}_{#1}
}

\NewDocumentCommand{\paramarchiclass}{m m}{
    \Theta_{#1,\mathbf{#2}}\xspace
}
\NewDocumentCommand{\archi}{m m}{
    (#1,\mathbf{#2})\xspace
}
\makeatletter
\DeclareRobustCommand\onedot{\futurelet\@let@token\@onedot}
\def\@onedot{\ifx\@let@token.\else.\null\fi\xspace}
\def\ie{{i.e}\onedot, }  
\def\eg{{e.g}\onedot, }  

\newtheoremstyle{mythm}
  {\topsep}
  {10pt}
  {\upshape}
  {0pt}
  {\bfseries}
  {. ---}
  { }
  {\thmname{#1}\thmnumber{ #2}\thmnote{ (#3)}}

\theoremstyle{mythm}%
\newtheorem{meta-thm}{Méta-théorème}%
\newtheorem{thm}{Theorem}[section]%
\newtheorem{prop}{Proposition}[section]
\newtheorem{cor}{Corollary}[section]%
\newtheorem{contrib}{Contribution}%
\newtheorem{lemma}{Lemma}[section]%
\newtheorem{example}{Example}[section]%
\newtheorem{informal}{Informal Theorem}[section]%

\newtheorem{definition}{Definition}[section]%
\crefname{definition}{Definition}{Definitions}

\newtheorem{remark}{Remark}[section]%

\title{Approximation speed of quantized \textit{vs.} unquantized ReLU neural networks and beyond}

\author{Antoine Gonon, Nicolas Brisebarre, Rémi Gribonval, Elisa Riccietti
\thanks{
Authors are with Univ Lyon, Ens Lyon, UCBL, CNRS, Inria,  LIP, F-69342, LYON Cedex 07, France. This work has been partially presented during a poster session of the winter school "LMS Invited Lectures On The Mathematics Of Deep Learning"
at the Isaac Newton Institute, Cambridge, UK (\href{https://sites.google.com/view/lmslecturesmaths4dl/home}{https://sites.google.com/view/lmslecturesmaths4dl/home}). It has also been partially presented during a talk at the conference "Curve and Surfaces 2022", Arcachon, France (\href{https://cs2022.sciencesconf.org/}{https://cs2022.sciencesconf.org/}). These two presentations did not lead to publications (there were no proceedings for both meetings). This work was supported in part by the AllegroAssai ANR-19-CHIA-0009 and NuSCAP ANR-20-CE48-0014 projects of the French Agence Nationale de la Recherche.
}}

\begin{document}

\maketitle

\begin{abstract}
    We deal with two complementary questions about approximation properties of ReLU networks. First, we study how the uniform quantization of ReLU networks with real-valued weights impacts their approximation properties. We establish an upper-bound on the minimal number of bits per coordinate needed for uniformly quantized ReLU networks to keep the same polynomial asymptotic approximation speeds as unquantized ones. We also characterize the error of nearest-neighbour uniform quantization of ReLU networks. This is achieved using a new lower-bound on the Lipschitz constant of the map that associates the parameters of ReLU networks to their realization, and an upper-bound generalizing classical results.
    Second, we investigate when ReLU networks can be expected, or not, to have better approximation properties than other classical approximation families. Indeed, several approximation families share the following common limitation: their polynomial asymptotic approximation speed of any set is bounded from above by the encoding speed of this set. We introduce a new abstract property of approximation families, called $\infty$-encodability, which implies this upper-bound. Many classical approximation families, defined with dictionaries or ReLU networks, are shown to be $\infty$-encodable. This unifies and generalizes several situations where this upper-bound is known.
\end{abstract}


\section{Introduction}
\noindent Neural networks are used  with success in many applications to approximate functions. In line with the works~\cite{DeVore2020SurveyNNApprox, Elbrachter21DNNApproximationTheory, Grohs15OptimallySparseDataRep}, we are interested in understanding their approximation power in practice and in theory. We deal with two complementary questions. Regarding practical applications, a key question is to be able to compare approximation properties of unquantized versus quantized neural networks, \ie networks with arbitrary real weights versus networks whose weights are constrained to a prescribed finite set (\eg floats). The results obtained in this direction are described below in \Cref{contrib1}. A practical type of quantization on which we will focus is uniform quantization, \ie when the weights are only allowed to be in a finite subset of a uniform grid of the real line. Another important question is to better understand non-trivial situations where neural networks, quantized or not, can be expected (or not) to have better approximation properties than the best known approximation families\footnote{An approximation family is any (often non-decreasing) sequence $(\approxclass_M)_{M=1,2,\dots}$ of subsets of a metric space $(\F,d)$.}. We lay a framework that can be used to identify such situations, as described in \Cref{contrib2}.
\begin{contrib}[approximation with quantized networks]\label{contrib1}
    Consider the parameters $\theta\in\R^d$ of a ReLU neural network and denote $R_{\theta}\in L^p$ the associated function, called its \emph{realization}, cf. Definition \ref{def:rtheta}. We also say that $R_\theta$ is \emph{realized} by the network with parameters $\theta$. Consider a quantization scheme\footnote{A quantization scheme is a function with a finite image.} $Q:\R^d\to\R^d$ used to quantize $\theta$ (\eg $Q(\theta)$ has only float coordinates). We address three main questions.

    \paragraph{Quantization error} A first question is to study how the quantization error $\|R_{\theta}-R_{Q(\theta)}\|_{p}$ (see \cref{subsec:ApproxSettingNN} for the definition of $\|\cdot\|_p$) depends on the quantization scheme $Q$. The following result goes in that direction.
    \begin{informal}[see \Cref{thm:BoundsLipschitzConstant}]\label{infThm:BoundsLipsCste}
        Fix an architecture (see \Cref{def:ArchiNN}) of feedforward ReLU networks, \ie fix the number of layers, denoted by $L$, and their width. Denote $W$ the maximal width of the layers. Denote $\Theta(r)$ the collection of all parameters of such networks having their Euclidean norm bounded by $r\geq 1$. Consider $1\leq p \leq \infty$. Then, the Lipschitz constant $\Lips(W,L,r)$ of the map that associates the parameters $\theta\in\Theta(r)$ to their realization $ R_\theta\in L^p$ satisfies, with constants $c,c'>0$ only depending on the $L^p$ space:
        \[
            c' Lr^{L-1}\leq \Lips(W,L,r)\leq c WL^2r^{L-1}.
        \]
    \end{informal}
    To the best of our knowledge, the lower-bound given in this result is new, while the upper-bound generalizes classical results \cite[Thm 2.6]{Berner20GeneralizationErrorBlackScholes}\cite[Lem. 2]{Neyshabur18PACBayesSpectrallyNormalizedMarginBounds} to generic $L^p$ spaces and to more general constraints on the parameters. Thanks to \Cref{infThm:BoundsLipsCste}, the number of bits used by a quantization scheme can be related to the error of this scheme. Our second result does so in the case of a nearest-neighbour uniform quantization scheme.
    \begin{informal}[see \Cref{thm:UnifQuantifErrorSufficientCond} and \Cref{thm:UnifQuantifErrorNecessaryCond}]\label{infThm:characQuantif}
        Consider the same setting as in \Cref{infThm:BoundsLipsCste}. Fix a stepsize $\eta>0$ and a desired error $\eps>0$. Consider the uniform quantization scheme\footnote{$\floor{\cdot}$ is defined as $\floor{x}:=\max\{n\in\Z, n\leq x\}$ for every $x\in\R$ while $\ceil{\cdot}$ is defined as $\ceil{x}:=\min\{n\in\Z, n\geq x\}$ for every $x\in\R$.} $Q_\eta(x)=\floor{x/\eta}\eta$  applied coordinate-wise on the parameters $\theta$ of a ReLU network. Then, \emph{$\|R_\theta-R_{Q_\eta(\theta)}\|_{\infty}\leq \eps$ holds for every $\theta\in\Theta(r)$ if, and only if, the number of bits used to store each coordinate of $Q_\eta(\theta)$}, which is proportional to $\ln(1/\eta)$, \emph{is linear in the depth $L$}.
    \end{informal}
    The generality of our first result suggests that the second one can be generalized to other settings, this will be further discussed in \Cref{rmk:GenLem6.8}. As a consequence of our first result, we also prove (cf. \Cref{prop:Lem6.8}) that Lemma \MakeUppercase{\romannumeral 6}.8 in \cite{Elbrachter21DNNApproximationTheory}, which controls quantization errors of the type $\|R_\theta-R_{Q(\theta)}\|_p$ when $Q$ is a nearest-neighbour uniform quantization scheme, can be improved.

    \paragraph{Approximation error} Given a function $f\in L^p$, how "well" can $f$ be approximated by quantized ReLU networks? If the parameters $\theta$ of a ReLU network are known to approximate "well" $f$, then one can simply quantize $\theta$ via a quantization scheme $Q$ and write, using the triangle inequality: $\|f-R_{Q(\theta)}\|_p\leq \|f-R_{\theta}\|_p + \|R_{\theta}-R_{Q(\theta)}\|_p$. The results above can be used to control the quantization error $\|R_{\theta}-R_{Q(\theta)}\|_p$. Applying this to functions $f$ in $L^\infty$-Sobolev spaces, and using external work \cite{Yarotsky17approxUnitBallSobolevWithDNN} to guarantee the existence of parameters $\theta$ approximating $f$ "well", we recover Theorem 2 in \cite{Ding19exprQuantizedNN}, see \Cref{prop:Ding}. Much more generic applications can be envisioned, see \Cref{rmk:Approx+QuantifErr}.

    \paragraph{Polynomial asymptotic approximation speed} Consider a function $f\in L^p$ and a sequence of parameters $(\theta_M)_{M\in\N}$ (with $\N=\{1,2,\dots\}$). Can we design a sequence $(Q_M)$ of quantization schemes such that the realizations of the networks with quantized parameters $(Q_{M}(\theta_M))_{M\in\N}$ approximate the function $f$ at the same asymptotic polynomial rate, with $M$, as the unquantized parameters $(\theta_M)_{M\in\N}$? Using the triangle inequality $\|f-R_{Q_{M}(\theta_M)}\|_p\leq \|f-R_{\theta_M}\|_p + \|R_{\theta_M}-R_{Q_{M}(\theta_M)}\|_p$ for each integer $M$, it is sufficient to guarantee that $\|R_{\theta_M}-R_{Q_{M}(\theta_M)}\|_p$ decreases at the same polynomial asymptotic rate as $\|f-R_{\theta_M}\|_p$. Given a subset $\fctclass$ of a metric function space $(\F,d)$ and an approximation family $\approxclass=(\approxclass_M)_{M\in\N}$ in $\F$, \emph{the polynomial asymptotic approximation speed $\approximationspeed(\fctclass|\approxclass)$ of $\fctclass$ by $\approxclass$} \cite[Def. \MakeUppercase{\romannumeral 5}.2, Def. \MakeUppercase{\romannumeral 6}.1]{Elbrachter21DNNApproximationTheory}, called simply approximation speed in what follows, \emph{is the best polynomial rate at which all functions of $\fctclass$ are asymptotically approximated by $\approxclass$}:
    \[
        \approximationspeed(\fctclass|\approxclass) :=\sup\{\gamma \in \mathbb{R}, \sup_{f\in\fctclass}\inf_{\Phi\in\approxclass_M}d(f,\Phi) = \mathcal{O}_{M \to \infty}\left(M^{-\gamma}\right)\} \in [-\infty,+\infty],
    \]
    with the convention $ \approximationspeed(\fctclass|\approxclass)  = -\infty$ if the supremum is over an empty set. In the following result, we exhibit a sufficient number of bits per coordinate that guarantees that nearest-neighbour uniform quantization preserves approximation rates of approximation families defined with ReLU networks.
    \begin{informal}[see \Cref{thm:ApproxSpeedQuantifNN}]
        \label{informalthm:ApproxSpeedQuantifNN}
        Consider the approximation family $\approxclass=(\approxclass_M)_{M\in\N}$ in an arbitrary $L^p$ space, such that $\approxclass_M$ is the set of functions realized by ReLU networks with depth bounded by $L_M\in\N$, with parameters having at most $M$ non-zero coordinates and with Euclidean norm bounded by $r_M\geq 1$. For $\speed>0$, consider the $\speed$-\emph{uniformly quantized} sequence $\mQ(\approxclass|\speed):=(\mQ(\approxclass_M|\speed))_{M\in\N}$, where $\mQ_{M}(\approxclass_M|\speed)$ is the set of functions realized by ReLU networks as above, but with parameters uniformly quantized using the quantization scheme $Q_{\eta_M}(x)=\floor{x/\eta_M}\eta_M$ for a step size $\eta_M=M^{-\speed}\Lips(M,L_M,r_M)$. Then, the $\speed$-\emph{uniformly quantized} sequence $\mQ(\approxclass|\speed)$ has, on every set $\fctclass\subset L^p$, an approximation speed which is comparable to its unquantized version $\approxclass$:
        \[
            \begin{array}{ll}
                \approximationspeed(\fctclass|\mQ(\approxclass|\speed)) = \approximationspeed(\fctclass|\approxclass) & \text{if $\speed\geq \approximationspeed(\fctclass|\approxclass)$}, \\
                \approximationspeed(\fctclass|\mQ(\approxclass|\speed)) \geq \speed                                   & \text{otherwise}.
            \end{array}
        \]
    \end{informal}
    This theorem leads to explicit conditions on the number of bits per coordinate that guarantee quantized ReLU networks to have the same approximation speeds as unquantized ones, see \Cref{ex:CompSpeedQuantizedNN}. In the proof, approximation speeds are matched by (i) taking unquantized parameters that (almost) achieve the unquantized approximation speed and (ii) quantizing these parameters with a sufficiently large number of bits in order to preserve the approximation speed. Smarter (but computationally more challenging) quantization schemes can be envisioned, such as directly picking the best quantized parameters to approximate the function. If the budget for the number of bits per coordinate is larger than the one given in \Cref{informalthm:ApproxSpeedQuantifNN}, then even the smartest quantization scheme will not beat approach $(i)+(ii)$ in terms of polynomial approximation speed (but it can still have better constants/log-terms etc.). Indeed, $(i)+(ii)$ already yields the same approximation speeds for quantized networks as unquantized ones, and quantized networks cannot do better than unquantized ones. An open question is: what is the minimum number of bits per coordinate needed to keep the same approximation speeds? We partially answer this question: \Cref{informalthm:ApproxSpeedQuantifNN} gives an upper-bound.
\end{contrib}
\begin{contrib}[unified and generic framework for a relation between approximation and information theory]\label{contrib2}
    We investigate generic settings where the approximation speed of a set $\fctclass$ by an approximation family is bounded from above by the encoding speed of $\fctclass$. The encoding speed is an informatic theoretic complexity that measures the best polynomial asymptotic rate at which the number of balls needed to cover the considered set grows as the radius of the balls goes to zero. Given a subset $\fctclass$ of a metric space $(\F,d)$ and $\eps>0$, a finite subset $X\subset\fctclass$ is called an \emph{$\eps$-covering} of $\fctclass$ if:
    \begin{equation}\label{def:CovNumber}
        \fctclass\subset\bigcup_{x\in X}B_d(x,\eps),
    \end{equation}
    where $B_d(x,\eps)$ denotes the closed ball of $\fctclass$, with respect to the metric $d$, centered in $x$ and with radius $\eps$. The \emph{covering number} $N(\fctclass,d,\eps)$ is the minimal size of an $\eps$-covering of $\fctclass$, with the convention that $N(\fctclass,d,\eps) = + \infty$ if there is no such covering. The \emph{metric entropy} is defined by $H(\fctclass,d,\eps):= \log_2(N(\fctclass, d, \eps))$. The encoding speed of $\fctclass$ is defined \cite[Def. \MakeUppercase{\romannumeral 4}.1]{Elbrachter21DNNApproximationTheory} as:
    \begin{equation}\label{def: encoding speed}
        \encodingspeed(\fctclass):=\sup\left\{\gamma >0, H(\fctclass,d,\eps) = \mathcal{O}_{\eps \to 0}(\eps^{-1/\speed})\right\},
    \end{equation}
    with the convention that $\encodingspeed(\fctclass) = 0$ if the supremum is over an empty set. The encoding speed is known for many $\fctclass$'s, see \cite[Table 1]{Elbrachter21DNNApproximationTheory}. Consider an approximation family $\approxclass=(\approxclass_M)_{M\in\N}$ and a set $\fctclass$. If $\approxclass_M$ approximates "well" $\fctclass$ (measured by $\approximationspeed(\fctclass|\approxclass)$), then one can use balls covering $\approxclass_M$ to cover $\fctclass$. This simple observation is at the origin of the following inequality, known in several situations, for instance when $\approxclass$ is defined with dictionaries \cite[Thm. \MakeUppercase{\romannumeral 5}.3]{Elbrachter21DNNApproximationTheory}\cite[Thm. 5.24]{Grohs15OptimallySparseDataRep} or ReLU networks \cite[Thm. \MakeUppercase{\romannumeral 6}.4]{Elbrachter21DNNApproximationTheory}:
    \begin{equation}\label{ineq: main ineq}
        \approximationspeed(\fctclass|\approxclass)\leq\encodingspeed(\fctclass).
    \end{equation}
    Inequality \eqref{ineq: main ineq} happens to be an equality in various cases, see \cite[Table 1]{Elbrachter21DNNApproximationTheory}. In \Cref{def:GammaEncodability}, we introduce an abstract property of approximation families $\approxclass=(\approxclass_M)_{M\in\N}$, called $\speed$-encodability, for $\gamma>0$, that measures how well each $\approxclass_M$ can be covered by balls. Roughly speaking, $\approxclass$ is $\gamma$-encodable if asymptotically in $M$, the set $\approxclass_M$ can be covered with nearly of the order of $M$ balls of radius $M^{-\speed}$. We say that it is $\infty$-encodable if it is $\gamma$-encodable for every $\gamma>0$. Our next result is to show that $\speed$-encodability can be used to understand several situations where \eqref{ineq: main ineq} holds.
    \begin{informal}[see \Cref{thm:unifying_thm}]\label{informalthm:unifying_thm}
        Consider an approximation family $\approxclass$ and $\speed>0$. If $\approxclass$ is $\speed$-encodable then for every set $\fctclass$:
        \[
            \min(\approximationspeed(\fctclass|\approxclass), \speed)\leq \encodingspeed(\fctclass).
        \]
    \end{informal}
    Note that when an approximation family is $\infty$-encodable, this result gives Inequality \eqref{ineq: main ineq}. Inequality \eqref{ineq: main ineq} is then of particular interest in order to bound the approximation speed $\approximationspeed(\fctclass|\approxclass)$ from above without having to look at all at the approximation properties of the set $\fctclass$ by the sequence $\approxclass$. Instead, we can study separately $\approxclass$ and establish at which speed it can be \emph{encoded}. This lays a framework that we use to unify and generalize several situations where Inequality \eqref{ineq: main ineq} is known \cite[Thm. \MakeUppercase{\romannumeral 5}.3, Thm. \MakeUppercase{\romannumeral 6}.4]{Elbrachter21DNNApproximationTheory}\cite[Thm. 5.24]{Grohs15OptimallySparseDataRep}. Indeed, we show that many approximation families $\approxclass=(\approxclass_M)_{M\in\N}$ are $\infty$-encodable: when $\approxclass_M$ contains $M$-terms linear combinations of the first $\text{poly}(M)$ elements\footnote{$M\mapsto\text{poly}(M)$ denotes a positive function that grows at most polynomially in $M$.} of a bounded dictionary, with boundedness conditions on the coefficients, or when $\approxclass$ is Lipschitz-parameterized (we say that an approximation family $(\approxclass_M)_{M\in\N}$ is Lipschitz-parameterized if there is a sequence $(B_M)_{M\in\N}$ of subsets of finite dimensional spaces and a sequence $(\varphi_M)_{M\in\N}$ of Lipschitz maps such that $\approxclass_M=\varphi_M(B_M)$ for every $M\in\N$), which is  the case for ReLU neural networks for which we identify "simple" sufficient conditions on the considered architectures for $\infty$-encodability to hold, see \cref{subsec:NNInfEnc}. Another consequence is that $\infty$-encodable approximation families $\approxclass$ defined with ReLU neural networks share a common upper bound on approximation rates with other classical approximation families that we prove to be $\infty$-encodable. In particular, given $\fctclass$, if an $\infty$-encodable sequence $\approxclass$ is known such that $\approximationspeed(\fctclass|\approxclass) =\encodingspeed(\fctclass)$ (examples of such situations can be found in \cite[Table 1]{Elbrachter21DNNApproximationTheory}), then no improved approximation rate using ReLU networks can be hoped for.
\end{contrib}

\textbf{Organization of the paper.} We recall the definition of feedforward ReLU neural network in \cref{subsec: def neural networks}. In \cref{subsec:ApproxSettingNN}, we describe the $L^p$ spaces in which approximation is considered. Bounds on the Lipschitz constants of the map that associates the parameters $\theta$ of ReLU networks to their realization $R_\theta\in L^p$ are given in \cref{sec:NNLipsParam}. The error of nearest-neighbour uniform quantization for ReLU networks is discussed in \cref{sec:UnifQuantifNN}. Approximation speeds of quantized ReLU networks are established in \cref{sec:ApproxSpeedQuantNN}. The notion of $\infty$-encodability is introduced in \cref{sec: encoding speeds vs approximation speeds}, before discussing its consequences on the relation between the approximation speed and the encoding speed. Examples of $\infty$-encodable sequences (\eg defined with dictionaries or ReLU networks) are then given in \cref{sec: gamma-encodable}. Sections \ref{sec: encoding speeds vs approximation speeds} and \ref{sec: gamma-encodable} are essentially independent of the others. We give some perspectives in \cref{sec:conc}. Some useful definitions, technical results and their proofs are gathered in the appendices.
\section{Preliminaries}\label{sec: preliminaries}
We recall the definition of ReLU neural networks, and characterize $L^{p}$ spaces containing their realizations.

\subsection{ReLU neural networks}\label{subsec: def neural networks}
\noindent A ReLU neural network is a parametric description of the alternate composition of affine maps between finite-dimensional spaces and of a non-linear function. The non-linearity consists of the so-called Rectified Linear Unit (ReLU) applied coordinate-wise.

\begin{definition}[ReLU: Rectified Linear Unit] The ReLU function $\rho$ is defined by \cite{arora2018understandingReLUnetworks}:
    \[
        \forall x\in\mathbb{R}, \;\rho(x) := \max(0,x).
    \]
    For $d\in\N$, its $d$-dimensional version consists of applying it coordinate-wise:
    \[
        \forall x\in\mathbb{R}^d,\; \rho(x) := (\rho(x_i))_{i=1\dots d}.
    \]
\end{definition}

\begin{figure}[ht]
    \centering
    \includegraphics[scale=0.3]{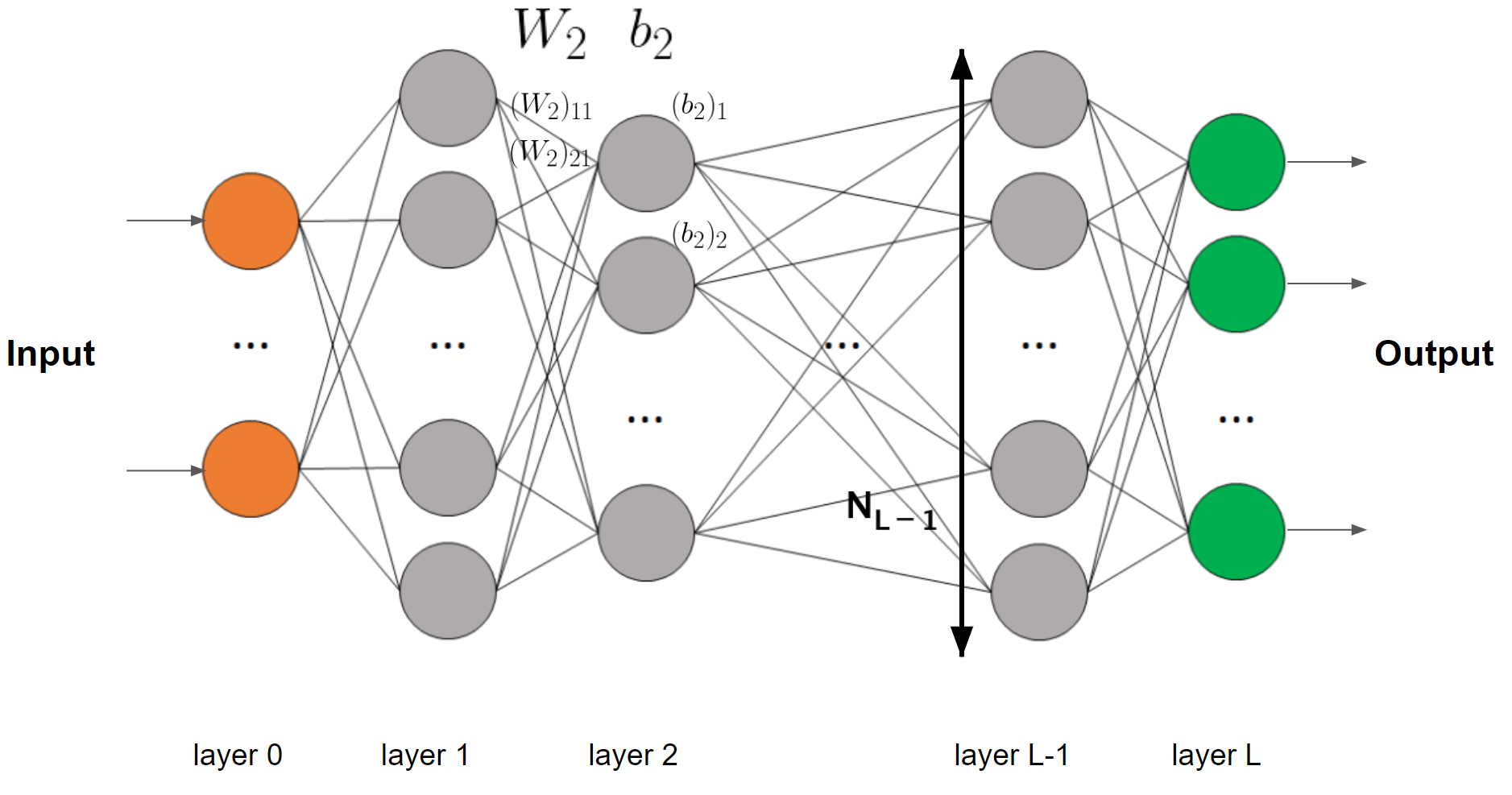}
    \caption{A neural network architecture can be seen as a directed graph. Neurons are represented by vertices, grouped by layers. For each neuron, there are edges going from this neuron to each neuron of the following layer. Coefficient $(i,j)$ of $W_\ell$ can be seen as the weight of the edge going from neuron $j$ of layer $\ell-1$ to neuron $i$ of layer $\ell$. Coefficient $i$ of $b_\ell$ can be seen as the weight of neuron $i$ of layer $\ell$.}
    \label{fig: réseau de neurones 1}
\end{figure}

\begin{definition}[Architecture of a neural network]\label{def:ArchiNN} An architecture of a neural network consists of a tuple $(L,\mathbf{N})$, with $L\in\N$ and $\mathbf{N}=(N_0,\dots,N_L)\in\N^{L+1}$. We then say that $L$ is the depth of the network. A network with such an architecture has $L+1$ layers of neurons, indexed from $\ell=0$ to $\ell=L$. Layer $\ell$ has $N_\ell$ neurons, we call $N_\ell$ the width of layer $\ell$. The width of the network is $W:=\max\limits_{\ell=0,\dots, L} N_\ell$. Layer $0$ is the input layer while layer $L$ is the output layer.
\end{definition}
An architecture can be represented as a graph, with a vertex for each neuron, and an edge between every pair of neurons within consecutive layers, see Figure \ref{fig: réseau de neurones 1} (thus, in this work, a layer consists of a set of neurons, not a set of edges).
\begin{definition}[Parameters associated to a network architecture] Let $(L,\mathbf{N})$ be an architecture. A parameter associated to this architecture consists of a vector $\theta=(W_1,\dots,W_L,b_1,\dots,b_L)$, with $W_\ell\in\R^{N_\ell\times N_{\ell-1}}$ and $b_\ell\in\R^{N_\ell}$. Such a parameter $\theta$ lives in the parameter space
    \begin{equation}\label{eq: def Theta L N et dim(L,N)}
        \begin{aligned}
            \paramarchiclass{L}{N} {} & := \mathbb{R}^{d_{\archi{L}{N}}},       \\
            d_{\archi{L}{N}} {}       & :=\sum_{\ell=1}^L N_\ell(N_{\ell-1}+1).
        \end{aligned}
    \end{equation}
\end{definition}
A parameter $\theta$ can be represented graphically: if neurons on layer $\ell$ are numbered from $1$ to $N_\ell$, then $(W_\ell)_{i,j}$ is the weight on the edge going from neuron $j$ of layer $\ell-1$ to neuron $i$ of layer $\ell$, while $(b_\ell)_i$ is the weight on neuron $i$ of layer $\ell$, see Figure \ref{fig: réseau de neurones 1}.
\begin{definition}[ReLU neural network and its realization]\label{def:rtheta}
    A ReLU neural network consists of an architecture $\archi{L}{N}$ and an associated parameter $\theta=(W_1,\dots,W_L,b_1,\dots,b_L)$. Its realization is the function denoted $R_{\param}:\mathbb{R}^{N_0}\rightarrow\mathbb{R}^{N_L}$, given by:
    \[
        \forall x\in\mathbb{R}^{N_0}, R_{\param}(x):=\Tilde{y}_L(x)
    \]
    with functions $y_\ell$ and $\Tilde{y}_\ell$ defined by induction on $\ell=1,\dots, L$:
    \begin{align*}
        y_0(x) {}               & =x,                           \\
        \Tilde{y}_{\ell+1}(x){} & =W_{\ell+1}y_\ell+b_{\ell+1}, \\
        y_{\ell+1}(x) {}        & =\rho(\Tilde{y}_{\ell+1}(x)).
    \end{align*}
    In words, the input $x$ goes through each layer sequentially, and when it goes from layer $\ell$ to $\ell+1$, it first goes through an affine transformation, of linear part $W_{\ell+1}$ and constant part $b_{\ell+1}$, then it goes through the ReLU function $\rho$ applied coordinate-wise (except on the last layer where the ReLU function is not applied).
\end{definition}
\subsection{Considered functional approximation setting}\label{subsec:ApproxSettingNN}
\noindent We consider $L^p$ spaces that contain all functions realized by ReLU networks (or equivalently all piecewise affine functions). We record the characterization of such spaces in \Cref{lem:CharacLpSpaces} (these are the  $L^p$ spaces for which \eqref{ineq: fct affine dans F} holds true) since we could not find it stated elsewhere. First let us introduce our notations for $L^p$ spaces. Let $d_{\textrm{in}},d_{\textrm{out}}\in\mathbb{N}$ be input and output dimensions, $p\in[1,\infty]$ be an exponent,  $\Omega\subset\mathbb{R}^{d_{\textrm{in}}}$ be the input domain and $\mu$ be a (non-negative) measure on $\Omega$. Given a norm $\|\cdot\|$ on $\mathbb{R}^{d_{\textrm{out}}}$, we define for every measurable function $f:\Omega\rightarrow\mathbb{R}^{d_{\textrm{out}}}$:
\begin{align*}
    \|f\|_{p,\|\cdot\|} :=
    \begin{cases}
        \left(\int_{x\in\Omega} \|f(x)\|^p\mathrm{d}\mu(x)\right)^{\frac{1}{p}} & \mbox{ if } p<\infty, \\
        \esssup\limits_{x\in\Omega} \|f(x)\|                                    & \mbox{ if } p=\infty.
    \end{cases}
\end{align*}
We consider approximation in the space $L^p(\Omega\to(\R^{d_{\textrm{out}}},\|\cdot\|),\mu)$ consisting of all measurable functions $f$ from $\Omega$ to $\R^{d_{\textrm{out}}}$ such that $\|f\|_{p,\|\cdot\|}<\infty$, quotiented by the relation ``being equal almost everywhere''. This is a Banach space with respect to the norm $\|\cdot\|_{p,\|\cdot\|}$. By the equivalence of norms in $\mathbb{R}^{d_{\textrm{out}}}$, this Banach space is independent of the choice of norm $\|\cdot\|$ on $\R^{d_{\textrm{out}}}$, and (for a given $p$) all norms $\|\cdot\|_{p,\|\cdot\|}$ are equivalent. In light of this fact we will simply denote it
$L^p(\Omega\to \R^{d_{\textrm{out}}},\mu)$, or even abbreviate it as $L^{p}$. We also denote $\|\cdot\|_{p} := \|\cdot\|_{p,\|\cdot\|_{\infty}}$. We will stress the dependence on the norm $\|\cdot\|$ when it plays a role, such as in \Cref{thm:BoundsLipschitzConstant}.

We now state a necessary and sufficient condition on $\Omega\subset\R^{d_{\textrm{in}}}$ and $\mu$ so that all functions realized by a ReLU neural network with input dimension $d_{\textrm{in}}$ and output dimension $d_{\textrm{out}}$ are in $L^p(\Omega\to \R^{d_{\textrm{out}}},\mu)$. The proof can be found in appendix \ref{app:CharacLpSpaces}.
\begin{lemma}\label{lem:CharacLpSpaces}
    Consider an exponent $p \in [1,\infty]$, a dimension $d_{\textrm{in}}$, a domain $\Omega\subset\mathbb{R}^{d_{\textrm{in}}}$, and a measure $\mu$ on $\Omega$. Define
    \[
        C_{p}(\Omega,\mu) :=
        \begin{cases}
            \left(\int_{x\in\Omega} (\|x\|_{\infty}+1)^p\mathrm{d}\mu(x)\right)^{1/p} & \mbox{ if } p<\infty, \\
            \esssup\limits_{x\in\Omega} \|x\|_{\infty}                                & \mbox{ if } p=\infty.
        \end{cases}
    \]
    The condition
    \begin{align}\label{ineq: fct affine dans F}
        C_{p}(\Omega,\mu) < \infty
    \end{align}
    is equivalent to: for every architecture $\archi{L}{N}$ with $N_0=d_{\textrm{in}}$ the realizations of ReLU networks satisfy:
    \[
        \forall\param\in\paramarchiclass{L}{N}, R_{\param}\in L^p(\Omega\to\R^{N_{L}},\mu),
    \]
    where $\paramarchiclass{L}{N}$ is defined in Equation~\eqref{eq: def Theta L N et dim(L,N)} and $N_{L}$ is the width of the output layer (\Cref{def:ArchiNN}).
\end{lemma}
Note that the condition $C_{p}(\Omega,\mu) < \infty$ holds in particular for every $p\in [1,\infty]$ when the input domain is bounded and $\mu$ is the Lebesgue measure.
\section{Lipschitz parameterization of ReLU neural networks}\label{sec:NNLipsParam}
\noindent It is known that some sets of functions realized by ReLU networks are Lipschitz-parameterized\footnote{A set is  Lipschitz-parameterized if it is the image by a Lipschitz map of a subset of a finite dimensional space.} \cite[Rmk. 9.1]{DeVore2020SurveyNNApprox}. In \Cref{thm:BoundsLipschitzConstant}, we give lower- and upper-bounds on the Lipschitz constant depending on the depth, the width and the weight's magnitude of the considered networks. To the best of our knowledge, the lower-bound is new, while the upper-bound generalizes similar upper-bounds established in specific cases \cite[Thm 2.6]{Berner20GeneralizationErrorBlackScholes}\cite[Lem. 2]{Neyshabur18PACBayesSpectrallyNormalizedMarginBounds} as discussed below.
These bounds will be useful in the next sections to understand how quantization of ReLU networks harms approximation error.
\begin{definition} \label{def: classe de param associee a une archi et G} (Parameter set $\paramarchiclass{L}{N}^q(r)$)
    Given an architecture $\archi{L}{N}$ and the set of associated parameters $\paramarchiclass{L}{N}$ (see Equation~\eqref{eq: def Theta L N et dim(L,N)}) we define for each $r\geq 0$ and $q \in [1,\infty]$ (the notation $\vertiii{\cdot}$ refers to the operator norm and is defined in appendix \ref{app: norms}):
    \[
        \paramarchiclass{L}{N}^q(r) := \{\param = (W_1,\dots,W_L,b_1,\dots,b_L) \in\paramarchiclass{L}{N}:
        \vertiii{W_\ell}_q,\|b_\ell\|_q\leqslant r, \ell=1,\dots ,L\}.
    \]
\end{definition}
Note that what will play a crucial role in what follows is the Lipschitz constant of the functions realized by the parameters in $\paramarchiclass{L}{N}^q(r)$. This Lipschitz constant is bounded by $r^L$ in the setup of \Cref{def: classe de param associee a une archi et G}. We do not enforce directly a global constraint on the Lipschitz constant since, to the best of our knowledge, there is no better practical way to enforce this constraint than by enforcing each $W_\ell$ and $b_\ell$ to have small norms. A more realistic situation thus corresponds to parameters $\theta$ with each $W_\ell$ and $b_\ell$ bounded for some norms, which is what reflects the definition of the set of parameters $\paramarchiclass{L}{N}^q(r)$ in \Cref{def: classe de param associee a une archi et G}.
\begin{remark}\label{rmk: parameter set with Frobenius or max norm} Instead of constraints on the operator norms, we may encounter constraints on the Frobenius or the max-norm. Let $r\geq 0$, and let $\archi{L}{N}$ be an architecture. Define by $W:=\max\limits_{\ell=0,\dots,L} N_\ell$ the width of the network. Denote $\|M\|_F=(\sum_{i,j} M_{i,j}^2)^{1/2}$ the Frobenius norm of a matrix $M$ and $\|M\|_{\max}=\max_{i,j}|M_{i,j}|$ the max-norm (to be distinguished from $\vertiii{M}_\infty$ the operator norm defined in appendix \ref{app: norms}), and define $\paramarchiclass{L}{N}^F(r)$ (resp. $\paramarchiclass{L}{N}^{\max}(r)$) the set of all $\param = (W_1,\dots,W_L,b_1,\dots,b_L)\in\paramarchiclass{L}{N}$ such that for every $\ell=1,\dots ,L$:
    \[
        \max\left(\|W_\ell\|_F, \|b_\ell\|_2\right) \leq r \text{ (resp. } \max\left(\|W_\ell\|_{\max}, \|b_\ell\|_\infty\right)\leq r).
    \]
    By standard results about equivalence of norms (see \eg \eqref{eq:OperatorNormEquivalence} in the appendix) it holds for every $q\in[1,\infty]$:
    \begin{align*}
        \paramarchiclass{L}{N}^F(r) {}      & \subset \paramarchiclass{L}{N}^2(r),                                            &
        \paramarchiclass{L}{N}^{\max}(r) {} & \subset \paramarchiclass{L}{N}^q(Wr) \subset \paramarchiclass{L}{N}^{\max}(Wr).
    \end{align*}
\end{remark}
Given an architecture $\archi{L}{N}$, we now give bounds on the Lipschitz constant of the map associating the parameters to their realization: $\theta\in\paramarchiclass{L}{N}^q(r)\mapsto R_\theta\in L^p$ . The proof is in appendix \ref{app: proof Lips constant}.
\begin{thm}\label{thm:BoundsLipschitzConstant} Consider $d_{\textrm{in}}, d_{\textrm{out}}\in\N$, $\Omega\subset\R^{d_{\textrm{in}}}$, $\mu$ a measure on $\Omega$ satisfying \eqref{ineq: fct affine dans F}, $\|\cdot\|$ a norm on $\R^{d_{\textrm{out}}}$, $p,q\in[1,\infty]$, and the space $\F := L^p(\Omega\to (\R^{d_{\textrm{out}}},\|\cdot\|),\mu)$. Then there exists a constant $c>0$ such that for every architecture $\archi{L}{N}$ with $N_0=d_{\textrm{in}}$ and $N_L=d_{\textrm{out}}$, and every $r\geqslant 1$, denoting by $W:=\max\limits_{\ell=0,\dots,L} N_\ell$ the width of the architecture, the map $\param\in\paramarchiclass{L}{N}^q(r)\mapsto R_{\param}\in L^p$ for ReLU networks satisfies
    \begin{equation}\label{eq:LipsConstDNN}
        \|R_{\param}-R_{\param'}\|_{p,\|\cdot\|} \leqslant cWL^2r^{L-1}\|\param-\param'\|_{\infty}\quad\text{ for all } \param, \param'\in\paramarchiclass{L}{N}^q(r).
    \end{equation}
    In particular, with $\mu$ the Lebesgue measure on $\Omega=[-D,D]^d$ for some $D>0$, this holds with:
    \begin{itemize}
        \item $c:=Dd^{1/q}+1$ if $p=\infty$ and $\|\cdot\| = \|\cdot\|_{q}$;
        \item $c:=(D+1)(2D)^{d/p}$ if $\|\cdot\|=\|\cdot\|_q=\|\cdot\|_\infty$.
    \end{itemize}

    Conversely, if $\Omega \subseteq \mathbb{R}_{+}^{d_{\textrm{in}}}$ (where $\R_+:=$), $\|\cdot\|=\|\cdot\|_q$ and $p=\infty$ then there exists a constant $c'>0$ independent of the architecture, such that, for every $\eps>0$, we can exhibit parameters $\param,\param'$ satisfying
    \begin{equation}\label{eq:LipsConstDNNConverse}
        \|R_{\param}-R_{\param'}\|_{p,\|\cdot\|} \geq (1-\eps)c'Lr^{L-1}\|\param-\param'\|_{\infty}.
    \end{equation}
    This converse result also holds for $1\leq p< \infty$ under the additional assumption that $N_{0} = \min_{0 \leq \ell \leq L} N_{\ell}$, \ie that all layers are at least as wide as the input layer.
\end{thm}
It is an open question whether the extra factor $WL$ in \eqref{eq:LipsConstDNN} compared to~\eqref{eq:LipsConstDNNConverse} can be improved, and whether the converse result for $p<\infty$ also holds without the additional assumption. Note that the condition $r\geq 1$ in \Cref{thm:BoundsLipschitzConstant} is reasonable since the realization of every parameter $\param\in\paramarchiclass{L}{N}^q(r)$ is a function $R_{\param}$ which is $r^L$-Lipschitz with respect to the $q$-norm on the input and output spaces. Constraining $r<1$ would lead to "very" flat functions, essentially constant, when $L$ is large. Vice-versa, the stability of a concrete numerical implementation of a neural network probably requires it to have a Lipschitz constant somehow bounded by the format used to represent numbers. Such considerations would probably lead to consider $r^{L} \leq C$ for some constant $C$, \ie $1 \leq r \leq C^{1/L}$.

Here is a list of immediate extensions of \Cref{thm:BoundsLipschitzConstant}:
\begin{itemize}
    \item \emph{Arbitrary Lipschitz activation: }\Cref{thm:BoundsLipschitzConstant} can be extended to the case where the ReLU activation function is replaced by any Lipschitz activation function.
    \item \emph{Pooling-operation: }\Cref{thm:BoundsLipschitzConstant} does not change if we add standard (max- or average-) pooling operations between some layers since they are $1$-Lipschitz.
    \item \emph{Arbitrary $s$-norm on the parameters: }since for every exponent $s\in[1,\infty]$, it holds $\|\cdot\|_\infty\leq \|\cdot\|_s$, \Cref{thm:BoundsLipschitzConstant} yields a bound on the Lipschitz constant with arbitrary $s$-norm on the parameter space.
    \item \emph{Generalization error bound: }in the context of learning, for a loss $\ell(\hat{y},y)$ that is a Lipschitz function of $\hat{y}$ with respect to some norm $\|\cdot\|$ on the support of a distribution $\P$, the excess risk $\E_{(x,y)\sim \P}(\ell(R_{\theta}(x),y)-\ell(R_{\theta'}(x),y))$ can be bounded from above by $\E_{(x,y)\sim \P}(\|R_{\theta}(x)-R_{\theta'}(x)\|)$, which in turn can be bounded using \Cref{thm:BoundsLipschitzConstant}. In particular, this is the case when $\P$ is supported on a compact set and $\ell(\hat{y},y)$ is continously differentiable in $\hat{y}$.
    \item \emph{Skip connections and convolutional layers: }one can also exploit \Cref{thm:BoundsLipschitzConstant} to networks with skip connections and/or convolutional layers, since they can be rewritten as networks with fully-connected layers. This rewriting can however   artificially inflate the widths of the networks and is unlikely to give sharp bounds. It is left to further work whether an extension of \Cref{thm:BoundsLipschitzConstant} with improved tailored bounds may be obtained in these settings.
\end{itemize}
\begin{remark}[Related works]\label{rmk:RelatedWorksBoundLipsCste}
    The fact that some sets of functions realized by ReLU neural networks are Lipschitz-parameterized is already known \cite[Rmk. 9.1]{DeVore2020SurveyNNApprox}. To our knowledge, the lower-bound in \Cref{thm:BoundsLipschitzConstant} is new. However the upper-bound is already known in several specific situations: at least for $d_{\textrm{out}}=1$, $L^\infty([0,1]^{d_{\textrm{in}}})$ with the Lebesgue measure, and $q=\max$ \cite[Thm 2.6]{Berner20GeneralizationErrorBlackScholes} as well as $p=\infty$, $q=F$, and $\|\cdot\|=\|\cdot\|_2$ \cite[Lem. 2]{Neyshabur18PACBayesSpectrallyNormalizedMarginBounds}. \Cref{thm:BoundsLipschitzConstant} shows that this upper-bound holds true more generally for general constraints on the parameters (arbitrary $q\in[1,\infty])$ and arbitrary $p\in[1,\infty]$ and $(\Omega,\mu)$ satisfying condition \eqref{ineq: fct affine dans F} \ie in any $L^p$ space that contains all the functions realized by ReLU neural networks. Let us also mention that \Cref{thm:BoundsLipschitzConstant} is based on \Cref{lem: ineq norm(rparam-rparam') and optimality} (appendix \ref{app:optimality}), and this lemma is a straightforward generalization of a known inequality for $q=\infty$ (see for instance \cite[Eq. (3.12)]{Bolcskei2018optimal} or \cite[Eq. (37)]{Elbrachter21DNNApproximationTheory}) to arbitrary $q\in[1,\infty]$. We prove that the inequality established in \Cref{lem: ineq norm(rparam-rparam') and optimality} is optimal. To our knowledge, even in the case $q=\infty$, the optimality has not been discussed yet in the literature.
\end{remark}

\section{Nearest-neighbour uniform quantization of ReLU neural networks}\label{sec:UnifQuantifNN}
\noindent In \cref{subsec:QuantError}, we characterize the error of nearest-neighbour uniform quantization of ReLU networks in $L^\infty$, recovering and improving Lemma \MakeUppercase{\romannumeral 6}.8 in \cite{Elbrachter21DNNApproximationTheory}. In \cref{subsec:Approx+QuantError}, we show that controlling the error of nearest-neighbour uniform quantization schemes leads to recover existing results \cite[Thm. 2]{Ding19exprQuantizedNN} on function approximation by quantized ReLU networks.
\subsection{Control of the $L^\infty$ quantization error $\|R_\theta-R_{Q(\theta)}\|_\infty$}\label{subsec:QuantError}
\noindent The following lemma is a direct consequence of \Cref{thm:BoundsLipschitzConstant}.
\begin{lemma}\label{lem:ErrorQuantifEta}
    Consider a domain $[-D,D]^d$. Fix an architecture $\archi{L}{N}=(L,(N_0,\dots,N_L))$ with width $W:=\max\limits_{\ell=0,\dots, L}N_\ell$, a bound $r\geq 1$ on the norm of the parameters, and an exponent $q\in[1,\infty]$. Given $\eta>0$, let $Q:\paramarchiclass{L}{N}\to \paramarchiclass{L}{N}$ (recall that $\paramarchiclass{L}{N}$ is the set of parameters associated with the architecture $\archi{L}{N}$, see \Cref{def:ArchiNN}) be such that $\|Q(\param)-\param\|_\infty\leq \eta$ for every parameter $\param\in\paramarchiclass{L}{N}^{q}(r)$. Consider a subset $\Theta\subset \paramarchiclass{L}{N}^{q}(r)$. Let $r'\geq 1$. Assume that:
    \begin{equation}\label{eq:QuantifPreservesNorm}
        Q(\param)\in\paramarchiclass{L}{N}^{q}(r'), \quad\forall \theta\in\Theta.
    \end{equation}
    Consider $\eps>0$, and $0<\eta\leq \eps\left(cWL^2(r')^{L-1}\right)^{-1}$, where $c:= D d^{1/q}+1$. Then, it holds:
    \begin{equation}\label{eq:EtaUnifQuantEpsDist}
        \max\limits_{\param\in\Theta}\
        \max\limits_{x\in[-D,D]^d} \|R_{\theta}(x)-R_{Q(\theta)}(x)\|_q\leq \eps.
    \end{equation}
\end{lemma}
\begin{proof}[Proof of \Cref{lem:ErrorQuantifEta}]
    Fix $\param\in\Theta$. Under assumption \eqref{eq:QuantifPreservesNorm}, it holds $Q(\param)\in\paramarchiclass{L}{N}^{q}(r')$. This means that we can apply \Cref{thm:BoundsLipschitzConstant} with $p=\infty$, $\|\cdot\|=\|\cdot\|_q$ and with the specific constant $c=Dd^{1/q}+1$. In this situation the essential supremum over $x\in[-D,D]^d$ in \Cref{thm:BoundsLipschitzConstant} is actually a maximum. This yields~\eqref{eq:EtaUnifQuantEpsDist} when $0<\eta\leq \eps\left(cWL^2(r')^{L-1}\right)^{-1}$.
\end{proof}
Under mild assumptions on the error $\eps$, Property \eqref{eq:QuantifPreservesNorm} holds for $r'=2r$. This leads to the following theorem.
\begin{thm}\label{thm:UnifQuantifErrorSufficientCond}
    In the same setting as in \Cref{lem:ErrorQuantifEta}, consider $0<\eps<cL^2(2r)^{L-1}$. If $0<\eta\leq \eps\left(cWL^2(2r)^{L-1}\right)^{-1}$, then \eqref{eq:EtaUnifQuantEpsDist} holds true.
\end{thm}
\begin{proof}[Proof of \Cref{thm:UnifQuantifErrorSufficientCond}]
    Fix $\theta=(W_1,\dots,W_L,b_1,\dots,b_L)\in\paramarchiclass{L}{N}^{q}(r)$. Assume that $0<\eps<cL^2(2r)^{L-1}$ and $0<\eta \leq \eps\left(cWL^2(2r)^{L-1}\right)^{-1}$. We want to prove that $Q(\theta)\in\paramarchiclass{L}{N}^{q}(2r)$. By assumption on $\eta$ and $\eps$, $0<\eta \leq \eps\left(cWL^2(2r)^{L-1}\right)^{-1}\leq 1/W \leq r/W$. Note that for a matrix $M$ with input and output dimensions bounded by $W$, it holds $\vertiii{M}_q\leq W\|M\|_{\max}$ , see \eqref{eq:OperatorNormEquivalence}. This guarantees that for every layer $\ell=1,\dots,L$, it holds $\vertiii{W_\ell-Q(W_\ell)}_q\leq W\|W_\ell-Q(W_\ell)\|_{\max}\leq W\eta\leq r$ and $\|b_\ell - Q(b_\ell)\|_q\leq W^{1/q}\|b_\ell-Q(b_\ell)\|_\infty\leq W \eta \leq r$ so that by the triangle inequality $Q(\param)\in\paramarchiclass{L}{N}^{q}(2r)$. Then, \eqref{eq:EtaUnifQuantEpsDist} follows from \Cref{lem:ErrorQuantifEta}.
\end{proof}
When $Q(x):=Q_\eta(x):=\floor{x/\eta}\eta$, we now establish a \emph{necessary} condition for \eqref{eq:EtaUnifQuantEpsDist} to hold, that almost matches the sufficient condition of \Cref{thm:UnifQuantifErrorSufficientCond}. This is obtained thanks to the almost matching lower- and upper-bounds of \Cref{thm:BoundsLipschitzConstant}. The proof is in appendix \ref{app:CharacUnifQuantif}.
\begin{thm}\label{thm:UnifQuantifErrorNecessaryCond}
    In the same setting as in \Cref{lem:ErrorQuantifEta}, consider the function $Q:=Q_\eta$ that acts coordinatewise on vectors and such that for every $x\in\R$, $Q_\eta(x)= \floor{x/\eta}\eta$. Define $N_{\min} := \min\limits_{0 \leq \ell \leq L} N_{\ell}$ and $c':=DN_{\min}^{1/q}$. If $\eps,\eta>0$ are such that \eqref{eq:EtaUnifQuantEpsDist} holds true for $\Theta := \paramarchiclass{L}{N}^{q}(r)$, then $\min(r,\eta)\leq \frac{\eps}{c'r^{L-1}}$. In particular, if $\eps < c' r^{L}$ then $\eta\leq \frac{\eps}{c'r^{L-1}}$.
\end{thm}
Note that \Cref{thm:UnifQuantifErrorSufficientCond} can be applied for every $\eps \in (0,1)$ since $cL^2(2r)^{L-1}>1$.  Similarly, if the domain $[-D,D]^{d}$ is large enough ($D \geq 1$) then $c' r^{L} > 1$ and \Cref{thm:UnifQuantifErrorNecessaryCond} yields that whenever \eqref{eq:EtaUnifQuantEpsDist} holds true for some $\eps \in (0,1)$ and $\eta>0$ we must have $\eta \leq \frac{\eps}{c' r^{L-1}}$.
\begin{remark}\label{rmk:BitsLinearGrowthInDepth}
    With $\eta>0$ and $Q_{\eta}$ the function from \Cref{thm:UnifQuantifErrorNecessaryCond}, the number of bits needed to store one coordinate of $Q_\eta(\theta)$ is proportional to $\ln(1/\eta)$. We just saw that if $D \geq 1$ and \eqref{eq:EtaUnifQuantEpsDist} is satisfied with $\eps \in (0,1)$, $\eta>0$, then $\eta$ must be exponentially small in $L$ (as soon as $r>1$). This means that the number of bits per coordinate must at least grow linearly with the network depth $L$ to ensure that the worst-case quantization error over networks in $\paramarchiclass{L}{N}^{q}(r)$ is controlled. This is essentially due to the fact that there are realizations of parameters in $\paramarchiclass{L}{N}^q(r)$ that are functions with Lipschitz constant \emph{equal} to $r^L$. More optimistic bounds can be envisioned under stronger assumptions on the set of parameters or on the network's architecture.
\end{remark}
Another direct consequence of \Cref{lem:ErrorQuantifEta} is the following proposition, which is proved in appendix \ref{app:CharacUnifQuantif} and yields an improvement of Lemma \MakeUppercase{\romannumeral 6}.8 in \cite{Elbrachter21DNNApproximationTheory}.
\begin{prop}[extension of {\cite[Lem. \MakeUppercase{\romannumeral 6}.8]{Elbrachter21DNNApproximationTheory}}]\label{prop:Lem6.8}
    Consider $\archi{L}{N}$ an architecture with input dimension $d_{\textrm{in}}$, output dimension $d_{\textrm{out}}$ and $L\geq 2$ layers. Consider the space $\F = L^\infty([-D,D]^{d_{\textrm{in}}}\to(\R^{d_{\textrm{out}}},\|\cdot\|_\infty), \mu)$ with $\mu$ the Lebesgue measure.

    Consider $\eps\in(0,1/2)$ and $\theta\in\paramarchiclass{L}{N}$. Denote $W=\max_{\ell=0,\dots, L} N_\ell$ the width of the architecture $\archi{L}{N}$. Let $k \geq 0$ be the smallest integer such that
    $\theta\in\paramarchiclass{L}{N}^{\max}(\eps^{-k})$ and $\max(W,L)\leq \eps^{-k}$, \ie $k = \lceil \log_2 \max(\|\theta\|_{\infty}, W, L)\\ /\log_2(1/\eps)\rceil$. For every integer $m \geq 2kL + k + 1 + \log_2(\ceil{D})$, the weights of $\theta$ can be rounded up to a closest point in $\eta\Z\cap[-\eps^{-k},\eps^{-k}]$ with $\eta:=2^{-m\ceil{\log_2(\eps^{-1})}} \leq \eps^{m}$ to obtain $Q_{\eta}(\theta)\in \paramarchiclass{L}{N}^{\max}(\eps^{-k})\cap (\eta\Z)^{d_{\archi{L}{N}}}$
    that satifies:
    \[
        \|R_\theta-R_{Q_{\eta}(\theta)}\|_\infty\leq \eps.
    \]
\end{prop}
Let us check that  \Cref{prop:Lem6.8} indeed implies the result of Elbr\"achter et al.\footnote{Lemma \MakeUppercase{\romannumeral 6}.8 in \cite{Elbrachter21DNNApproximationTheory} is stated for networks having at most $\eps^{-k}$ non-zero weights. Given such a network, we can always remove neurons having only zero incoming and outcoming weights. This gives another network, with the same realization, but with a width $W\leq \eps^{-k}$ and a depth $L\leq \eps^{-k}$. Then, \Cref{prop:Lem6.8} applies to this new network. (R1-m20)}~\cite[Lem. \MakeUppercase{\romannumeral 6}.8]{Elbrachter21DNNApproximationTheory}. First, since $\max(W,L)\geq L \geq 2$, it holds $k\geq 1$. Thus, for $L \geq 2$, we have $k(L-1) \geq 1$ so that $3kL \geq 2kL+k+1$ and it is thus sufficient to take $m \geq 3kL+\log_2(\ceil{D})$ (which is the sufficient condition given in \cite[Lem. \MakeUppercase{\romannumeral 6}.8]{Elbrachter21DNNApproximationTheory}). Note however the improved (slower) growth of $m$ with $L$ in the sufficient condition of \Cref{prop:Lem6.8} compared to \cite[Lem. \MakeUppercase{\romannumeral 6}.8]{Elbrachter21DNNApproximationTheory}.
\begin{remark}\label{rmk:GenLem6.8}
    More generally, given bounds on the sparsity (\ie number of nonzero entries), on the magnitude of the network weights, and an $\textit{arbitrary}$ $p\in[1,\infty]$, \Cref{thm:BoundsLipschitzConstant} can be used to find an appropriate step size that guarantees that a uniform quantization of the considered network is within error $\eps>0$ in $L^p$.
\end{remark}
\subsection{Control of the approximation error of a function by quantized networks $\|f-R_{Q(\theta)}\|_\infty$}\label{subsec:Approx+QuantError}
\noindent Given a function $f$, parameters $\theta$, and a quantization function $Q$, a simple triangle inequality yields
$\|f-R_{Q(\theta)}\|_\infty \leq \|f-R_\theta\|_\infty + \|R_\theta-R_{Q(\theta)}\|_\infty$. \Cref{thm:BoundsLipschitzConstant} controls the quantization error $\|R_\theta-R_{Q(\theta)}\|_\infty$ for nearest-neighbour uniform quantization schemes. If, in addition, information about the approximation error $\|f-R_\theta\|_\infty$ is available, then we can deduce a bound on the approximation error of $f$ by quantized networks.

We apply this simple observation in the case of functions $f$ in an $L^\infty$-Sobolev space to recover a special case of Theorem 2 in \cite{Ding19exprQuantizedNN} (the other cases can be recovered by combining this special case with Proposition 3 in \cite{Ding19exprQuantizedNN}). The proof is in appendix \ref{app:CharacUnifQuantif}.

Let $n\in\N$ and consider $\mathcal{W}^{n,\infty}([0,1]^d)$, the Sobolev space of real-valued functions on $[0,1]^d$ that are in $L^\infty$ as well as their weak derivatives up to order $n$ (given $\mathbf{n}:=(n_1,\dots,n_d)\in\N^d$, the associated weak-derivative of a function $f$ is denoted $D^\mathbf{n}f$ if it exists). The norm on $\mathcal{W}^{n,\infty}([0,1]^d)$ is given by:
\[
    \|f\|_{\mathcal{W}^{n,\infty}([0,1]^d)} := \max\limits_{\substack{\mathbf{n}:=(n_1,\dots,n_d)\in\N^d\\ \sum_in_i\leq n}} \esssup\limits_{x\in[0,1]^d}|D^\mathbf{n}f(x)|.
\]
\begin{prop}[{\cite[Thm. 2]{Ding19exprQuantizedNN}}]\label{prop:Ding} Let $\fctclass_{n,d}$ be the unit ball of $\mathcal{W}^{n,\infty}([0,1]^d)$. There exists a constant $c>0$ depending only on $n$ and $d$ such that for every $\eps\in(0,1)$, there exists $\eta>0$ satisfying $\ln(1/\eta)\leq c\ln^2(1/\eps)$ and a ReLU network architecture that can approximate every function $f\in\fctclass_{n,d}$ within error $\eps>0$ in $L^\infty([0,1]^d)$ using weights in $\eta\Z$,  with depth bounded by $c\ln(1/\eps)$, a number of  weights at most equal to $c\eps^{-d/n}\ln(1/\eps)$, and with a total number of bits (used to store the network weights) bounded by $c\eps^{-d/n}\ln^3(1/\eps)$.
\end{prop}
\begin{remark}\label{rmk:Approx+QuantifErr}
    Compared to Theorem 2 of \cite{Ding19exprQuantizedNN}, \Cref{thm:BoundsLipschitzConstant} can also be used to establish similar results, not only for a function $f$ in the unit ball of an $L^\infty$-Sobolev space, but for every $f\in L^p$ ($1\leq p\leq \infty$) \emph{as soon as it is known how to approximate $f$ with unquantized ReLU networks, with explicit bounds on the growth of their depth, width and weight's magnitude}. For instance, such bounds are known for  Hölder spaces \cite{Ohn19ApproxHolderSpaceWithReLUNets}, classifier functions in $L^2$ \cite{Petersen18ApproxClassifierFunctInL2WithReLUNets} and Besov spaces \cite{Suzuki19ApproxBesovSpaceWithReLUNets}. The same argument also applies for networks with arbitrary Lipschitz activation (such as the sigmoid function) for which an analog of \Cref{thm:BoundsLipschitzConstant} can be derived, and for which we know how to approximate "smooth" functions \cite[Table 1]{Guhring21ApproxReLUNNGeneralActivation}.
\end{remark}
\section{Approximation speeds of quantized \emph{vs.} unquantized ReLU neural networks}\label{sec:ApproxSpeedQuantNN}
\noindent Consider a function $f$ and a sequence of parameters $(\theta_M)_{M\in\N}$ such that $\|f-R_{\theta_M}\|_p$ goes to zero as $M$ goes to infinity. Given a nearest-neighbour uniform quantization scheme $Q_M$ with a step size that depends on $M$, we saw in the previous sections how to control the quantization error $\|R_{\theta_M}-R_{Q(\theta_M)}\|_p$. We used the latter to control $\|f-R_{Q(\theta_M)}\|_p\leq \|f-R_{\theta_M}\|_p + \|R_{\theta_M}-R_{Q(\theta_M)}\|_p$, in specific situations such as when $f$ is in an $L^\infty$-Sobolev space, see \cref{subsec:Approx+QuantError}. In this section, we give sufficient conditions on the step size used for quantization with $Q_M$ to guarantee that the quantization error $\|R_{\theta_M}-R_{Q(\theta_M)}\|_p$ decreases, with $M$, at the same asymptotic polynomial rate as the approximation error $\|f-R_{\theta_M}\|_p$. This leads to an explicit sufficient number of bits, depending on the growth with $M$ of the architecture of the parameters $\theta_M$, that guarantees that quantized ReLU networks have the same approximation speeds as unquantized ones, see \Cref{ex:CompSpeedQuantizedNN}. First, we define the considered approximation families and their uniformly quantized version. Notations are somewhat cumbersome but necessary to introduce the different types of constraints on the considered architectures (depth and width constraints) and on the parameters (sparsity and norm constraints).
\begin{definition}[Sequence of sets of architectures]\label{def:SeqSetsArchi}
    Consider $d_{\textrm{in}}, d_{\textrm{out}}\in\N$ and $(L_M)_{M\in\N}\in\N^\N$. For each $M\in\N$ define $\archiclass{M}$, the set of architectures with input dimension $d_{\textrm{in}}$, output dimension $d_{\textrm{out}}$, depth bounded by $L_M$ and widths of the hidden layers bounded by $M$:
    \begin{equation}\label{eq:DefArchiM}
        \archiclass{M} :=
        \{(L,(N_0,\dots,N_L)) :  L,N_0,\dots, N_L\in\mathbb{N}, L\leqslant L_M, N_0=d_{\textrm{in}}, N_L=d_{\textrm{out}},\nonumber  N_\ell\leqslant M, \ell=1,\dots ,L-1\}. \nonumber
    \end{equation}
    For every $M\in\N$ and every architecture $\archi{L}{N}\in\archiclass{M}$, define $S^M_{\archi{L}{N}}$ as the set of all supports $S\subset \{0,1\}^{d_{\archi{L}{N}}}$ \emph{of cardinality at most $M$}, used to constrain the non-zero entries of a vector $\param$ with architecture $\archi{L}{N}$.
\end{definition}
The width constraint $N_{\ell}\leq M$ in the definition of the architectures in $\archiclass{M}$ is written for clarity but is superfluous in what follows, given that the realization of a network $\theta$ (with arbitrary activation function and an architecture of arbitrary width) with at most $M$ nonzero coefficients can always be written as the realization of a parameter $\theta'$ on a ``pruned'' architecture where every hidden layer has width $N_{\ell} \leq M$.
\begin{definition}[ReLU networks approximation family]\label{def:ApproxFamilyNN}
    Consider a sequence $(\archiclass{M})_{M\in\N}$ of sets of architectures as in \Cref{def:SeqSetsArchi}. Consider $L^p(\Omega\to(\R^{d_{\textrm{out}}},\|\cdot\|),\mu)$ satisfying \eqref{ineq: fct affine dans F}, $q\in[1,\infty]\cup\{F,\max\}$ ($F$ and $\max$ refer to the Frobenius norm and the max-norm, see \Cref{rmk: parameter set with Frobenius or max norm}), and a sequence $(r_M)_{M\in\N}$ of real numbers such that $r_M\geq 1$. Define the approximation family $\nnclass:=(\nnclass_M)_{M\in\N}$ of sets $\nnclass_M\subset L^p(\Omega\to(\R^{d_{\textrm{out}}},\|\cdot\|),\mu)$ of realizations of ReLU neural networks with an architecture $\archi{L}{N}\in\archiclass{M}$ and parameters in $\paramarchiclass{L}{N}^q(r_M)$:
    \[
        \nnclass_M :=\bigcup\limits_{\archi{L}{N}\in\archiclass{M}}\bigcup\limits_{S\in S^M_{\archi{L}{N}}}R_{\paramarchiclass{L}{N}^q(r_M), S}
    \]
    where for any parameter set $\Theta$ and support $S$ we denote $R_{\Theta, S} := \{R_{\theta}, \theta \in \Theta\text{ supported on } S\}$.
\end{definition}
\begin{definition}[Quantization of parameters on a uniform grid]\label{def:QThetaetar} Given a set of parameters $\Theta$ associated to an architecture $\archi{L}{N}$, we define, for every $\eta>0$ and $r\in(0,\infty]$, the quantized version $\mQ(\Theta,\eta,r)$ of $\Theta$ on a bounded uniform grid $(\eta\Z\cap[-r,r])^{d_{\archi{L}{N}}}$:
    \[
        \mQ(\Theta,\eta,r):=\Theta \cap (\eta\Z\cap[-r,r])^{d_{\archi{L}{N}}}.
    \]

    \Cref{def:QThetaetar} will essentially be used when $\Theta$ is a ball. In that case, $Q(\Theta,\eta,r)$ is not empty provided that $r$ and $\Theta$ are sufficiently large compared to the step size $\eta$.
\end{definition}
\begin{definition}[Quantized version of \Cref{def:ApproxFamilyNN}]\label{def:QuantifApproxFamilyNN}
    Consider an approximation family $\nnclass:=(\nnclass_M)_{M\in\N}$ as in \Cref{def:ApproxFamilyNN}, with the associated sequences $(r_M)_{M\in\N}, (L_M)_{M\in\N}$ and $q\in[1,\infty]\cup\{F,\max\}$. For every $M\in\N$, define
    \[
        \Lips(M,q) := \begin{cases}
            \max(d_{\textrm{in}}, d_{\textrm{out}},M)L_M^2r_M^{L_M-1}   & \text{if } q\in[1,\infty]\cup\{F\}, \\
            \Lips(M,2)\max(d_{\textrm{in}}, d_{\textrm{out}},M)^{L_M-1} & \text{for }q=\max,
        \end{cases}
    \]
    and observe that $\Lips(M,q) \geq 1$.
    Given any $\speed>0$ define for every $M\in\N$ the step size $\eta_M=\eta_{M}(\gamma,q):=(M^{\gamma}\Lips(M,q))^{-1}$. The \emph{$\speed$-uniformly quantized} version $\mQ(\nnclass|\speed):=(\mQ_{M}(\nnclass_M|\speed))_{M\in\N}$ of $\nnclass$ is
    \begin{equation*}
        \mQ_{M}(\nnclass_M|\speed) = \bigcup\limits_{\archi{L}{N}\in\archiclass{M}}\bigcup\limits_{S\in S^M_{\archi{L}{N}}} R_{\mQ(\paramarchiclass{L}{N}^q(r_M), \eta_M, r_M), S}
    \end{equation*}
    with $\archiclass{M}$ and $S^M_{\archi{L}{N}}$ the families of architectures and supports (cf \Cref{def:SeqSetsArchi}) associated to $\mathcal{N}_{M}$, see \Cref{def:ApproxFamilyNN} where we recall that the notation $R_{\Theta,S}$ given a support $S$ is also introduced.
\end{definition}
In general, Lipschitz-parameterized approximation families can be uniformly quantized into sequences having comparable approximations speeds, if a step size sufficiently small is chosen. \Cref{thm:ApproxSpeedQuantifNN} deals only with the case of Lipschitz-parameterized approximation families we are interested in: ReLU neural networks. The upper-bound on the Lipschitz constant established in \Cref{thm:BoundsLipschitzConstant} yields explicit conditions on the growth of the depth and the weight's magnitude, that guarantee that the \emph{$\speed$-uniformly quantized} sequence $Q(\nnclass|\speed)$ has, on every set $\fctclass\subset L^p$, an approximation speed which is comparable to its unquantized version $\nnclass$. The proof of \Cref{thm:ApproxSpeedQuantifNN} is in appendix \ref{app:LipsParamUnifQuant}.
\begin{thm}\label{thm:ApproxSpeedQuantifNN}
    Consider the context of \Cref{def:QuantifApproxFamilyNN}. Then, for every $\speed>0$, the $\speed$-\emph{uniformly quantized} sequence $\mQ(\nnclass|\speed):=(\mQ_{M}(\nnclass_M|\speed))_{M\in\N}$ has, on every (non-empty) set $\fctclass\subset L^p$, an approximation speed comparable to the unquantized one $\nnclass$:
    \begin{equation}\label{eq:QuantSpeedNN}
        \begin{array}{ll}
            \approximationspeed(\fctclass|\mQ(\nnclass|\speed)) = \approximationspeed(\fctclass|\nnclass) & \text{if $\speed\geq \approximationspeed(\fctclass|\nnclass)$}, \\
            \approximationspeed(\fctclass|\mQ(\nnclass|\speed)) \geq \speed                               & \text{otherwise}.
        \end{array}
    \end{equation}
\end{thm}
We will see in the following sections that the approximation speed $\approximationspeed(\fctclass|\nnclass)$ can be bounded from above by a quantity denoted $\encodingspeed(\fctclass)$, the latter quantity being known for several classical sets $\fctclass$ (see \cite[Table 1]{Elbrachter21DNNApproximationTheory}). In such a situation, this guides the choice of $\speed$ to define a concrete $\speed$-quantized sequence in the context of \Cref{thm:ApproxSpeedQuantifNN}. Indeed, considering  $\fctclass\subset\F$ a classical function class for which the quantity $\encodingspeed(\fctclass)$ is known, choosing $\speed\geq\encodingspeed(\fctclass)$ is sufficient to ensure that $\speed \geq  \approximationspeed(\fctclass|\nnclass)$. Vice-versa, among all such $\speed$, choosing the smallest one $\speed = \encodingspeed(\fctclass)$ is probably the best choice to yield the largest possible step sizes $\eta_M$ and the best concrete compromise.
\begin{example}[Comparable approximation speeds with controlled growth of the number of bits]\label{ex:CompSpeedQuantizedNN}
    Let $q\in[1,\infty]\cup\{F,\max\}$ be an exponent and $\pi$ be a positive polynomial and consider $\nnclass_M^\pi$ the set of functions parameterized by a ReLU neural network with arbitrary architecture $\archi{L}{N}$ with depth bounded by $\pi(\log M)$, with at most $M$ non-zero parameters and with parameters in $\paramarchiclass{L}{N}^q(\pi(M))$. For every $\speed>0$, there exists a constant $c(\gamma)>0$ such that the $\speed$-uniformly quantized sequence $Q(\nnclass_M^\pi|\speed)$ of $\nnclass_M^\pi$ is obtained with step size $\eta_M=\mathcal{O}_{M\to\infty}(M^{-c(\gamma)\log M})$, \ie using $\mathcal{O}_{M\to\infty}((\log M)^2)$ bits per weight. \Cref{thm:ApproxSpeedQuantifNN} guarantees that this quantized sequence still has approximation speeds comparable to $\nnclass_M^\pi$. In the same setup, if we assume in addition that the depths $L_M$ are uniformly bounded in $M$, then for every $\speed>0$, a step size $\eta_M=\mathcal{O}_{M\to\infty}(M^{-c(\gamma)})$ (\ie $\mathcal{O}_{M\to\infty}(\log M)$ bits per parameter) suffices to get comparable speeds as in Equation~\eqref{eq:QuantSpeedNN}.
\end{example}

\section{Encoding speeds \emph{vs} approximation speeds}
\label{sec: encoding speeds vs approximation speeds}
\noindent We now investigate a fundamental limitation of many approximations families (including ReLU networks): the approximation speed of a set by an approximation family cannot be greater than the encoding speed of this set (see \eqref{ineq: main ineq}). \Cref{sec: encoding speeds vs approximation speeds} and \cref{sec: gamma-encodable} are essentially independent from the others. We introduce an abstract property of approximation families, called "encodability", in \Cref{def:GammaEncodability}. In \Cref{thm:unifying_thm}, we prove that every approximation family satisfying this encodability property must satisfy Inequality \eqref{ineq: main ineq}. As we will see in \cref{sec: gamma-encodable}, this lays a unified and generic framework that captures and recovers different known situations \cite[Thm. \MakeUppercase{\romannumeral 5}.3, Thm. \MakeUppercase{\romannumeral 6}.4]{Elbrachter21DNNApproximationTheory}\cite[Thm. 5.24]{Grohs15OptimallySparseDataRep}\cite[Prop. 11]{Kerkyacharian04Entropy} where \eqref{ineq: main ineq} holds.
\subsection{The notion of $\gamma$-encodability}
\noindent Let $\approxclass:=(\approxclass_M)_{M\in\N}$ be a sequence of non-empty subsets of a metric space $(\F,d)$. Let $\fctclass\subset\F$ and $\eps>0$. If $\approximationspeed(\fctclass|\approxclass)>0$, since $\approxclass$ approximates $\fctclass$ at speed $\approximationspeed(\fctclass|\approxclass)$, there exists a positive integer $M$ large enough such that every element $f\in\fctclass$ can be $\eps$-approximated (with respect to the metric $d$) by an element of $\approxclass_M$. Since $\approxclass_M$ can be $\eps$-covered (with respect to $d$) with $N(\approxclass_M, d, \eps)$ elements, $\fctclass$ can be $2\eps$-covered with $N(\approxclass_M, d, \eps)$ elements. Instances of this simple reasoning can be found in \cite[Thm. \MakeUppercase{\romannumeral 5}.3, Thm. \MakeUppercase{\romannumeral 6}.4]{Elbrachter21DNNApproximationTheory}\cite[Thm. 5.24]{Grohs15OptimallySparseDataRep}\cite[Prop. 11]{Kerkyacharian04Entropy}. This suggests the existence of a relation between the approximation speed $\approximationspeed(\fctclass|\approxclass)$ and the encoding speed $\encodingspeed(\fctclass)$ that depends on the growth with $M$ of the covering numbers of $\approxclass_M$.

We claim that a "reasonable" growth of the covering numbers of $\approxclass_M$  consists in a situation where, for some $\speed>0$, the set $\approxclass_M$ can be $M^{-\speed}$-covered with "roughly" $2^{M\log M}$ elements. Indeed, this covers the case where each element of $\approxclass_M$ can be described by $M$ parameters that can be stored with a number of bits per parameter that grows logarithmically in $M$. For instance if $\approxclass_M$ is a bounded set in dimension $M$ then it can be uniformly quantized along each dimension with a size step of order $M^{-\gamma}$, so that $\log M$ bits is roughly enough to encode each of the $M$ coordinates. This "reasonable" growth for the covering numbers of $\approxclass_M$ is formalized in \Cref{def:GammaEncodability}, and yields the simple relation $\min(\approximationspeed(\fctclass|\approxclass),\speed)\leq \encodingspeed(\fctclass)$ for every set $\fctclass\subset\F$, as shown in \Cref{thm:unifying_thm}.
\begin{definition}[$(\speed,h)$-encoding]\label{def:(gamma,h)-encoding}
    Let $(\F,d)$ be a metric space. Let $\approxclass:=(\approxclass_M)_{M\in\N}$ be an arbitrary sequence of (non-empty) subsets of $\F$. Let $\speed>0$ and $h>0$. A sequence $(\approxclass(\speed,h)_M)_{M\in\N}$ is said to be a \emph{$(\speed,h)$-encoding} of $\approxclass$ if there exist constants $c_1,c_2>0$ such that for every $M\in\N$, the set $\approxclass(\speed,h)_M$ is a $c_1M^{-\speed}$-covering of $\approxclass_M$ (recall \Cref{def:CovNumber}, in particular $\approxclass(\speed,h)_M$ must be a subset of $\approxclass_M$) of size satisfying $\log_2(|\approxclass(\speed,h)_M|)\leq c_2M^{1+h}$.
\end{definition}
The following definition captures a "reasonable" growth with $M$ of the covering numbers of $\approxclass_M$.
\begin{definition}[$\speed$-encodable $\approxclass$ in $(\F,d)$]\label{def:GammaEncodability} Let $(\F,d)$ be a metric space. Let $\approxclass:=(\approxclass_M)_{M\in\N}$ be an arbitrary sequence of (by default, non-empty) subsets of $\F$. Let $\speed>0$. We say that $\approxclass$ is \emph{$\speed$-encodable} in $(\F,d)$ if for every $h>0$, there exists a $(\speed,h)$-encoding of $\approxclass$. We say that $\approxclass$ is \emph{$\infty$-encodable} in $(\F,d)$ if it is $\speed$-encodable in $(\F,d)$ for all $\speed>0$. When the context is clear, we will omit the mention to $(\F,d)$.
\end{definition}
Note that if $\approxclass$ is $\speed$-encodable then it is $\speed'$-encodable for every $\speed'\leq\speed$. Several examples of $\infty$-encodable sequences are given in \cref{sec: gamma-encodable}, including classical approximation families defined with dictionaries or ReLU networks.
\subsection{The encoding speed as a universal upper bound for approximation speeds}
\noindent It is known that $\approximationspeed(\fctclass|\approxclass)\leq \encodingspeed(\fctclass)$ for various sets $\fctclass$ when $\approxclass$ is defined with neural networks \cite[Thm. \MakeUppercase{\romannumeral 6}.4]{Elbrachter21DNNApproximationTheory} or dictionaries \cite[Thm. \MakeUppercase{\romannumeral 5}.3]{Elbrachter21DNNApproximationTheory}\cite[Thm. 5.24]{Grohs15OptimallySparseDataRep}. The following proposition shows that $\infty$-encodability implies $\approximationspeed(\fctclass|\approxclass)\leq \encodingspeed(\fctclass)$. This settles a unified and generalized framework for the aforementioned known cases that implicitly use, one way or another, the $\infty$-encodability property, as we will detail in \cref{subsec:InfEncDict} and \cref{subsec:NNInfEnc}.
\begin{thm}\label{thm:unifying_thm}
    Consider $(\F,d)$ a metric space and $\approxclass:=(\approxclass_M)_{M\in\N}$ an arbitrary sequence of (non-empty) subsets of $\F$ which is $\speed$-encodable in $(\F,d)$, with $\gamma\in(0,\infty]$.
    Then for every (non-empty) $\fctclass\subset\F$:
    \[
        \min (\approximationspeed(\fctclass|\approxclass),\speed )\leq \encodingspeed(\fctclass).
    \]
\end{thm}
The proof of \Cref{thm:unifying_thm} is in appendix \ref{app:EncImpliesUnifyingBound}. We derive from \Cref{thm:unifying_thm} a generic lower bound on the encoding speed of the set of functions uniformly approximated at a given speed.
\begin{cor}\label{cor:EncSpeedBallApproxClass} Let $(\F,d)$ be a metric space. Consider $\speed\in(0,\infty]$ and $\approxclass:=(\approxclass_M)_{M\in\N}$ an arbitrary sequence of (non-empty) subsets of $\F$ which is $\speed$-encodable in $(\F,d)$. Consider $\alpha,\beta>0$ and $\mathcal{A}^\alpha(\F,\approxclass,\beta)$ the set of all $f\in\F$ such that
    $\sup_{M\geq 1} M^{\alpha}d(f,\approxclass_M)\leq \beta$. This set satisfies
    \[
        \encodingspeed(\mathcal{A}^\alpha(\F,\approxclass,\beta)) \geq \min(\alpha,\speed).
    \]
\end{cor}
\begin{proof}
    By the very definition of $\mathcal{A}^\alpha(\F,\approxclass,\beta)$, it holds $\approximationspeed(\mathcal{A}^\alpha(\F,\approxclass,\beta)|\approxclass)\geq \alpha$. \Cref{thm:unifying_thm} then gives the result.
\end{proof}
The reader may wonder about the role of $\beta$ in the above result, and whether a similar result can be achieved with $\mathcal{A}^\alpha(\F,\approxclass) := \cup_{\beta>0}\mathcal{A}^\alpha(\F,\approxclass,\beta)$. While this is left open, a related discussion after \Cref{prop:AppSpace} suggests this may not be possible without additional assumptions on $\approxclass$.

As an immediate corollary of \Cref{thm:unifying_thm}  we also obtain the following result.

\begin{cor}\label{cor: approximation speed bounded from above by encoding speed}
    Consider $\approxclass:=(\approxclass_M)_{M\in\N}$ an arbitrary sequence of (non-empty) subsets of a metric space $\F$ and a (non-empty) set $\fctclass\subset\F$. If $\approxclass$ is $\speed$-encodable for every $\speed<\approximationspeed(\fctclass|\approxclass)$ then:
    \[
        \approximationspeed(\fctclass|\approxclass)\leq\encodingspeed(\fctclass).
    \]
\end{cor}
\begin{proof}
    For every $\speed<\approximationspeed(\fctclass|\approxclass)$, since $\Sigma$ is $\gamma$-encodable,  we have $\speed = \min ( \approximationspeed(\fctclass|\approxclass), \gamma ) \leq\encodingspeed(\fctclass)$ by Proposition \ref{thm:unifying_thm}. Taking the supremum of such $\speed$, we get the inequality.
\end{proof}

As we will see in \cref{sec: gamma-encodable}, applying \Cref{cor: approximation speed bounded from above by encoding speed} to specific $\infty$-encodable sequences allows one to unify and generalize different cases where $\approximationspeed(\fctclass|\approxclass)\leq\encodingspeed(\fctclass)$ is known to hold \cite[Thm. \MakeUppercase{\romannumeral 5}.3, Thm. \MakeUppercase{\romannumeral 6}.4]{Elbrachter21DNNApproximationTheory}\cite[Thm. 5.24]{Grohs15OptimallySparseDataRep}.

Note that the quantity $\encodingspeed(\fctclass)$ is known in several cases, see \cite[Table 1]{Elbrachter21DNNApproximationTheory}. In the next section, we discuss concrete examples of $\infty$-encodable sequences $\approxclass$. For such a sequence $\approxclass$ and an arbitrary set $\fctclass$, independently of the adequation of $\Sigma$ and $\fctclass$, \Cref{cor: approximation speed bounded from above by encoding speed} automatically yields an upper bound for the approximation speed of $\fctclass$ by $\Sigma$.

In some situations, the converse of \Cref{cor: approximation speed bounded from above by encoding speed} can be established.
\begin{thm}
    Let $\fctclass$ be a (non-empty) subset of a metric space $(\F,d)$ and $\approxclass:=(\approxclass_M)_{M\in\N}$ a sequence of (non-empty) subsets of $\F$ such that  $\approxclass_M\subset \fctclass$ for every $M$ large enough. If $\min(\approximationspeed(\fctclass|\approxclass), \encodingspeed(\fctclass))>0$ then the sequence $\approxclass$ is $\speed$-encodable for each $0<\speed<\min(\approximationspeed(\fctclass|\approxclass), \encodingspeed(\fctclass))$. \\
    In particular, if $\approximationspeed(\fctclass|\approxclass)\leq \encodingspeed(\fctclass)$ then $\Sigma$ is $\speed$-encodable for every $0<\speed<\approximationspeed(\fctclass|\approxclass)$.
\end{thm}
\begin{proof}
    Consider $0<\speed<\min(\approximationspeed(\fctclass|\approxclass), \encodingspeed(\fctclass))$. By definition of $\approximationspeed(\fctclass|\approxclass)$, there exists a constant $c>0$ such that for every $f\in\fctclass$ and every $M\in\N$, there exists $\Phi_M(f)\in\approxclass_M$ such that $d(f,\Phi_M(f))\leq cM^{-\speed}$. Consider $\speed'>0$ such that $\speed<\speed'<\min(\approximationspeed(\fctclass|\approxclass), \encodingspeed(\fctclass))$. For $M\in\N$, define $\eps_M:=M^{-\speed}$.  By definition of $\encodingspeed(\fctclass)$, there is a constant $c'>0$ such that for every $\eps>0$, there exists an $\eps$-covering $\fctclass_\eps$ of $\fctclass$ of size satisfying $\log_2(|\fctclass_\eps|)\leq c'\eps^{-1/\speed'}$. For $M$ large enough, $\approxclass_M\subset\fctclass$, hence for every such $M$ and every $f\in\approxclass_M$, there exists $f_{\eps_{M}}\in\fctclass_{\eps_{M}}$ such that $d(f,f_{\eps_{M}})\leq \eps_{M}$. Using the triangle inequality, we obtain that for every $M$ large enough and every $f\in\approxclass_M$: $d(f,\Phi_M(f_{\eps_M}))\leq (1+c)M^{-\speed}$. This shows that $\Phi_M(\fctclass_{\eps_M})$ is a $(1+c)M^{-\speed}$-covering of $\approxclass_M$ of size satisfying $\log_2(|\Phi_M(\fctclass_{\eps_M})|)\leq c'M^{\speed/\speed'}$, with $\speed/\speed'<1$. This shows that $\approxclass$ is $\speed$-encodable. The rest of the claim follows.
\end{proof}
\section{Examples of $\infty$-encodable approximation families}\label{sec: gamma-encodable}
\noindent We now give several examples of $\infty$-encodable sequences $\approxclass$. We start with a gentle warmup in \cref{subsec:FiniteToLipsParamInfEnc}. It is proven that some sequences of balls (in the sense of the metric space $\F$) of increasing radius and dimension are $\infty$-encodable. Quite naturally, $\infty$-encodability is preserved under some Lipschitz transformation, as shown in \Cref{thm:LipsParamAreUnifQuant} in the specific case of $\infty$-encodable sequences of balls (this can be generalized to other $\infty$-encodable sequences, but this is not useful here). In \cref{subsec:InfEncDict}, we give examples of $\infty$-encodable sequences in the context of approximations with dictionaries, see \cref{subsec:InfEncDict}, showing that \Cref{thm:unifying_thm} unifies and generalizes Theorem \MakeUppercase{\romannumeral 5}.3 in \cite{Elbrachter21DNNApproximationTheory} and Theorem 5.24 in \cite{Grohs15OptimallySparseDataRep}. Finally, in \cref{subsec:NNInfEnc}, we give an example of an $\infty$-encodable approximation family defined with ReLU networks. Once again, \Cref{thm:unifying_thm} applied to this $\infty$-encodable sequence recovers a known result, see \Cref{ex:InfEncNN}.
\subsection{First examples of $\infty$-encodable sequences}\label{subsec:FiniteToLipsParamInfEnc}
\noindent This subsection is a gentle warmup, where basic examples of $\infty$-encodable sequences are given in order to manipulate the notion of encodability. Let $(\F,d)$ be a metric space and $c>0$. Let $\approxclass:=(\approxclass_M)_{M\in\N}$ be a sequence of sets $\approxclass_M\subset\F$ that can be \emph{covered} with $N_M=\mathcal{O}_{M\to\infty}(2^{cM\pi(\log M)})$ balls (with respect to the ambient metric space) centered in $\approxclass_M$ of radius $\eps_M=\mathcal{O}_{M\to\infty}(M^{-\speed})$. Since $\mathcal{O}_{M\to\infty}(2^{cM\pi(\log M)})=\mathcal{O}_{M\to\infty}(2^{M^{1+h}})$ for every $h>0$, it is clear from the definition that $\approxclass$ is $\infty$-encodable. This is trivially the case when $\approxclass:=(\approxclass_M)_{M\in\N}$ is a sequence of finite sets $\approxclass_M\subset\F$ with at most $2^{cM \pi(\log M)}$ elements since each $\approxclass_M$ is an exact covering of itself. Another example consists of some sequences of balls (in the sense of the metric space $\F$) of increasing radius and dimension as described in the next lemma. The proof is in appendix \ref{app:FirstExInfEnc}.
\begin{lemma}\label{lem: some sequences of balls in finite dim are infinite-encodable}
    Consider $q\in[1,\infty]$, $(d_M)_{M\in\N}\in\N^\N$, $(r_M)_{M\in\N}$ a sequence of real numbers satisfying $r_M\geq 1$ and define $\approxclass:=(\approxclass_M)_{M\in\N}$, with $\approxclass_M:=B_{d_M,\|\cdot\|_q}(0,r_M)$ being the set of sequences of $\ell^q(\N)$ bounded by $r_M$ and supported in the first $d_M$ coordinates. Then, $\approxclass$ is either $\infty$-encodable in $\ell^q(\N)$ or it is never $\speed$-encodable in $\ell^q(\N)$, whatever $\speed>0$ is. Moreover, it is $\infty$-encodable if, and only if,
    \[
        d_M\left(\log_2(r_M)+1\right)=\mathcal{O}_{M\to\infty}(M^{1+h}), \qquad \forall h >0.
    \]
\end{lemma}
Quite naturally, $\infty$-encodability can be preserved under Lipschitz maps as shown in the following theorem. The proof is in appendix \ref{app:FirstExInfEnc}.
\begin{thm}\label{thm:LipsParamAreUnifQuant} Consider the same setting as in \Cref{lem: some sequences of balls in finite dim are infinite-encodable}. Consider also a sequence $\varphi:=(\varphi_M)_{M\in\N}$ of maps $\varphi_M:(\approxclass_{M}, \|\cdot\|_q)\to(\F,d)$ that are $\Lips(\varphi_M)$-Lipschitz for some constants $\Lips(\varphi_M)\geq 1$. Define $\varphi(\approxclass):=(\varphi_M(\approxclass_M))_{M\in\N}$. Assume that for every $h>0$:
    \begin{equation}\label{eq:LipCondition}
        d_M\left(\log_2(r_M)+\log_2(\Lips(\varphi_M))+1\right)=\mathcal{O}_{M\to\infty}(M^{1+h}).
    \end{equation}
    Then $\varphi(\approxclass)$ is $\infty$-encodable.
\end{thm}
\subsection{The case of dictionaries}\label{subsec:InfEncDict}
\noindent We now consider sequences $\approxclass$ defined with dictionaries. As detailed below, results of the literature \cite[Thm. 5.24]{Grohs15OptimallySparseDataRep}\cite[Prop. 11]{Kerkyacharian04Entropy} use arguments that implicitly prove $\speed$-encodability. Let us start with the case of approximation in Banach spaces as in \cite{Kerkyacharian04Entropy}. We only explicit the sequence used in \cite{Kerkyacharian04Entropy} which is $\speed$-encodable and we do not delve into more details as results of \cite{Kerkyacharian04Entropy} are out of scope of this paper. A part of the proof of \cite[Prop. 11]{Kerkyacharian04Entropy} consists of implicitly showing that some specific sequence $\approxclass^q$ is $s$-encodable, for $q$ and $s$ as described below in \Cref{prop:KerkyacharianEnc}, as shown in appendix \ref{app:DictNNInfEnc}. In particular, the setup of \Cref{prop:KerkyacharianEnc} applies when $\F$ is the $L^p$ space on $\mathbb{R}^d$ or $[0,1]^d$, $1<p<\infty$, and the basis $B$ is a compactly supported wavelet basis or associated wavelet-tensor product basis.
\begin{prop}\label{prop:KerkyacharianEnc}
    Let $\F$ be a Banach space with a basis $B=(e_i)_{i\in\N}$ satisfying $\sup_{i\in\N}\|e_i\|_\F<\infty$. Consider $p\in(0,\infty)$ and assume that $B$ satisfies the so-called $p$-Telmyakov property \cite[Def. 2]{Kerkyacharian04Entropy}, \ie assume that there exists $c>0$ such that for every finite subset $I$ of $\N$:
    \begin{equation}\label{eq:p-Telmyakov}
        \frac{1}{c}|I|^{1/p}\min_{i\in I}|c_i|\leq \|\sum_{i\in I}c_ie_i  \|_\F\leq c |I|^{1/p}\max_{i\in I}|c_i|, \quad \forall (c_i)_{i\in I}\in\R^I.
    \end{equation}
    Consider $0<q<p$. For every $M\in\N$, define\footnote{In terms of weak-$\ell^q$-space, the set $\approxclass^q_M$ is simply the set of linear combinations of elements of $B$ given by sequences $(c_i)_{i\in\N}$ in the closed unit ball of $\ell^{q,\infty}(\N)$ with zero coordinates outside the first $M$ ones.}:
    \[
        \approxclass^q_M := \left\{\sum_{i=1}^M c_ie_i, c_i\in\R, \sup_{0<\lambda<\infty}\lambda|\{i, |c_i|\geq \lambda\}|^{1/q}\leq 1\right\}.
    \]
    Define $s=\frac{1}{q}-\frac{1}{p}$. Then the sequence $\approxclass^q:=(\approxclass_M^q)_{M\in\N}$ is $s$-encodable in $\F$.
\end{prop}
In the case of Hilbert spaces, much more generic sequences than $\approxclass^q$ above are in fact $\infty$-encodable, as we now discuss. The $\infty$-encodability can be used to recover \cite[Thm. 5.24]{Grohs15OptimallySparseDataRep} (see \Cref{cor:Grohs}), and to generalize \Cref{cor:EncSpeedBallApproxClass} (see \Cref{prop:AppSpace}). Let $\F$ be a Hilbert space and $d$ be the metric associated to the norm on $\F$. A dictionary is, by definition \cite[Def. 5.19]{Grohs15OptimallySparseDataRep}, a subset $\D=(\phi_i)_{i\in\N}$ of $\F$ indexed by a countable set, which we assume to be $\N$ without loss of generality. The dictionary $\D$ can be used to approach elements of $\F$ by linear combinations of a growing number $M$ of its elements.
\begin{thm}\label{thm:GSInfEnc} Let $\F$ be a Hilbert space. Let $\D=(\phi_i)_{i\in\N}$ be a dictionary in $\F$, and $\pi:\N\to\N$ be a function with at most polynomial growth. For every $I\subset\N$, define $(\Tilde{\phi}^I_i)_{i\in I}$ as any orthonormalization of $(\phi_i)_{i\in I}$ (for instance we may consider the Gram-Schmidt orthonormalization). Define for every $M\in\N$ and $c>0$:
    \begin{align*}
        \approxclass^\pi_M             & := \left\{\sum_{i\in I} c_i\phi_i, I\subset\{1,\dots,\pi(M)\}, |I| \leq M, (c_i)_{i\in I}\in\R^I \right\},                                \\
        \tilde{\approxclass}^{\pi,c}_M & := \left \{\sum_{i\in I}\Tilde{c}_i\Tilde{\phi}^I_i,  I\subset\{1,\dots,\pi(M)\}, |I|\leq M, (\Tilde{c}_i)_{i\in I}\in[-c,c]^I \right \}.
    \end{align*}
    The sequence $\tilde{\approxclass}^{\pi,c}:=(\tilde{\approxclass}^{\pi,c}_M)_{M\in\N}$ is $\infty$-encodable in $(\F,d)$, and for every bounded set $\fctclass\subset\F$, it holds:
    \begin{equation}\label{eq:GSSpeed}
        \approximationspeed(\fctclass|\approxclass^\pi) = \max_{c>0}\approximationspeed(\fctclass|\tilde{\approxclass}^{\pi,c}).
    \end{equation}
\end{thm}
The proof of \Cref{thm:GSInfEnc} is in appendix \ref{app:DictNNInfEnc}. As a consequence of \Cref{thm:GSInfEnc}, one can recover \cite[Thm. 5.24]{Grohs15OptimallySparseDataRep} as we now describe.
\begin{cor}[{\cite[Thm. 5.24]{Grohs15OptimallySparseDataRep}}]\label{cor:Grohs}
    Let $(\F,d)$ be a Hilbert space and $\fctclass\subset\F$. Under the assumptions of \Cref{thm:GSInfEnc}, the sequence $\approxclass^\pi=(\approxclass^\pi_M)_{M\in\N}$ satisfies for every relatively compact\footnote{Recall that a set is relatively compact if its closure is compact. In particular, it must be totally bounded, and in particular bounded.} set $\fctclass$:
    \[
        \approximationspeed(\fctclass|\approxclass^\pi) \leq \encodingspeed(\fctclass).
    \]
\end{cor}
Actually, instead of stating the previous result with the approximation speed $\approximationspeed(\fctclass|\approxclass)$, Theorem 5.24 in \cite{Grohs15OptimallySparseDataRep} considers the following quantity \cite[Def. 5.23]{Grohs15OptimallySparseDataRep}:
\[
    \speed^*(\fctclass|\approxclass):=\sup\{\speed\in\R, \forall f\in\fctclass,\exists c>0,\forall M\in\N, d(f,\approxclass_M)\leq cM^{-\speed}\},
\]
which satisfies $\speed^*(\fctclass|\approxclass) \geq \approximationspeed(\fctclass|\approxclass)$ but generally differs from $\approximationspeed(\fctclass|\approxclass)$ since in the definition of $\approximationspeed(\fctclass|\approxclass)$, the implicit constant $c>0$ is not allowed to depend on $f\in\fctclass$. However, when $\fctclass$ is relatively compact (that is, its closure is compact), then $c>0$ can be chosen independently of $f$ \cite[Proof of Thm 5.24]{Grohs15OptimallySparseDataRep} so that the two quantities coincide. The proof of \Cref{cor:Grohs} that can be found below is essentially a rewriting in the formalism of \cref{sec: encoding speeds vs approximation speeds} of the original proof of Theorem 5.24 in \cite{Grohs15OptimallySparseDataRep}. The rewriting makes explicit the use of equality \eqref{eq:GSSpeed} and the $\infty$-encodability of the sequences $\tilde{\approxclass}^{\pi,c}$ for $c>0$, which are only implicitly used in the original proof.
\begin{proof}
    Since $\fctclass$ is relatively compact, it must be bounded so Equation~\eqref{eq:GSSpeed} of \Cref{thm:GSInfEnc} holds. For every $c>0$, \Cref{thm:unifying_thm} applied to $\tilde{\approxclass}^{\pi,c}$ of \Cref{thm:GSInfEnc}, which is $\infty$-encodable, shows that the right hand-side of Equation~\eqref{eq:GSSpeed} is bounded from above by $\encodingspeed(\fctclass)$. This yields the result.
\end{proof}
We also obtain a generic lower bound on the encoding speed of balls of approximation spaces \cite[Sec. 7.9]{DeVore1993} (also called maxisets \cite{Kerkyacharian2000}) with general dictionaries.
\begin{cor}\label{prop:AppSpace}
    Let $(\F,d)$ be a Hilbert space. Under the assumptions of \Cref{thm:GSInfEnc}, consider $\alpha,\beta>0$ and the set\footnote{This is the ball of radius $\beta$ of an approximation space \cite[Sec. 7.9]{DeVore1993}/maxiset\cite{Kerkyacharian2000}.}\footnote{Note that compared to the set in \Cref{cor:EncSpeedBallApproxClass}, we additionally require that $\|f\| \leq \beta$ so that $\mathcal{A}^\alpha(\F,\approxclass^\pi,\beta)$ is a bounded set and Equation~\eqref{eq:GSSpeed} of \Cref{thm:GSInfEnc} holds.} $\mathcal{A}^\alpha(\F,\approxclass^\pi,\beta)$  of all $f\in\F$ such that $\|f\| \leq \beta$ and $\sup_{M\geq 1} M^{\alpha}d(f,\approxclass_M)\leq \beta$. This set satisfies
    \[
        \encodingspeed(\mathcal{A}^\alpha(\F,\approxclass^\pi,\beta)) \geq \alpha.
    \]
\end{cor}

\Cref{prop:AppSpace} cannot be generalized to $\mathcal{A}^\alpha(\F, \approxclass^\pi):=\bigcup_{\beta> 0} \mathcal{A}^\alpha(\F, \approxclass^\pi, \beta)$: this set is homogeneous (stable by multiplication by any scalar), thus it cannot be encoded at any positive rate. Indeed, a positive encoding rate implies total boundedness of a set, whereas homogeneity implies that the set cannot be totally bounded (at least under the assumption that the metric is induced by a norm; there should, in general, be metrics with respect to which a homogeneous set may be totally bounded).

In some situations, the converse inequality $\encodingspeed(\mathcal{A}^\alpha(\F,\approxclass^\pi,\beta)) \leq \alpha$ can typically be proven by studying the existence of large enough packing sets of $\mathcal{A}^\alpha(\F,\approxclass^\pi,\beta)$, but this falls out of the scope of this paper. The reader can refer to \cite[Sec. 4]{Kerkyacharian04Entropy} for an example.
\begin{proof}[Proof of \Cref{prop:AppSpace}]
    By the very definition of $\fctclass := \mathcal{A}^\alpha(\F,\approxclass^\pi,\beta)$, this is a bounded set so Equation~\eqref{eq:GSSpeed} of \Cref{thm:GSInfEnc} holds. For every $c>0$, \Cref{thm:unifying_thm} applied to $\tilde{\approxclass}^{\pi,c}$ of \Cref{thm:GSInfEnc}, which is $\infty$-encodable, shows that the right hand-side of Equation~\eqref{eq:GSSpeed} is bounded from above by $\encodingspeed(\fctclass)$, so that
    \[
        \encodingspeed(\fctclass) \geq \max_{c>0} \approximationspeed(\fctclass|\tilde{\approxclass}^{\pi,c}) = \approximationspeed(\fctclass|\approxclass^{\pi}).
    \]
    Finally, again by definition of $\fctclass := \mathcal{A}^\alpha(\F,\approxclass^\pi,\beta)$, we have $\approximationspeed(\fctclass|\approxclass^\pi)\geq \alpha$.
\end{proof}

Note that if $\approxclass^\pi$ was $\speed$-encodable for some $\speed>0$ large enough then \Cref{cor:Grohs} would be a special case of \Cref{cor: approximation speed bounded from above by encoding speed} whereas \Cref{prop:AppSpace} would be a special case of \Cref{cor:EncSpeedBallApproxClass}. But in this situation, $\approxclass^\pi$ has no reason to be $\speed$-encodable, whatever $\speed>0$ is (since the dictionary is arbitrary and the coefficients of the linear combinations are not bounded). This shows that \Cref{cor: approximation speed bounded from above by encoding speed} and \Cref{cor:EncSpeedBallApproxClass} actually holds more generally for some sequences $\approxclass$ that are not $\speed$-encodable, whatever $\speed>0$ is, as soon as $\approxclass$ can be recovered as a limit of non-decreasing sequences $\approxclass^c$, $c>0$, that are $\speed$-encodable, in the sense that for every $M\in\N$, if $0<c\leq c'$ then $\approxclass_M^c\subseteq\approxclass_M^{c'}$ and $\approxclass_M=\cup_{c>0}\approxclass_M^c$.
\subsection{The case of ReLU networks}\label{subsec:NNInfEnc}
When $\approxclass$ is defined with ReLU feed-forward neural networks, we now explicitly study how the property of $\infty$-encodability depends on (bounds on) the neural network sparsity, depth, and weights. In particular, \Cref{prop:InfEncNN} establishes a "simple" explicit condition under which \Cref{thm:unifying_thm} generalizes Theorem \MakeUppercase{\romannumeral 6}.4 in \cite{Elbrachter21DNNApproximationTheory} to other type of constraints. \Cref{prop:InfEncNN} is proven in appendix \ref{app:DictNNInfEnc}.
\begin{prop}\label{prop:InfEncNN}
    Consider the context of \Cref{def:QuantifApproxFamilyNN} and assume that for every $h>0$, it holds:
    \begin{equation}\label{eq:NSCArchGrowth2}
        L_{M} M \left(1+\log_{2}(r_{M})\right) = \mathcal{O}_{M\to\infty}(M^{1+h}).
    \end{equation}
    Then the approximation family $\nnclass$ (\Cref{def:ApproxFamilyNN}) defined with ReLU networks is $\infty$-encodable.
\end{prop}
\begin{example}[{$\infty$-encodable sequences of \emph{sparse} neural networks - \cite[Thm. \MakeUppercase{\romannumeral 6}.4]{Elbrachter21DNNApproximationTheory}}]\label{ex:InfEncNN}
    Let $\pi$ be a positive polynomial and consider, as in Definition \MakeUppercase{\romannumeral 6}.2 of \cite{Elbrachter21DNNApproximationTheory}, $\nnclass_M^\pi$ the set of functions parameterized by a ReLU neural network with weights' amplitude bounded by $\pi(M)$, depth bounded by $\pi(\log M)$ and at most $M$ non-zero parameters. Assumption \eqref{eq:NSCArchGrowth2} holds since this corresponds to the case where $L_M\leq \pi(\log(M))$ and $1\leq r_M\leq \max(1,\pi(M))$. Then, \Cref{prop:InfEncNN} guarantees that $\nnclass^\pi:=(\nnclass^\pi_M)_{M\in\N}$ is $\infty$-encodable. Given \Cref{thm:unifying_thm}, the fact that $\nnclass^\pi$ is $\infty$-encodable gives $\approximationspeed(\fctclass|\nnclass^\pi)\leq\encodingspeed(\fctclass)$ for arbitrary $p\in[1,\infty]$ and arbitrary $\fctclass\subset L^p$. This is exactly Theorem \MakeUppercase{\romannumeral 6}.4 in \cite{Elbrachter21DNNApproximationTheory}.
\end{example}
\section{Conclusion}\label{sec:conc}
We now summarize our different contributions and discuss perspectives.

\noindent \textbf{Approximation with quantized ReLU networks} We characterized the error of simple uniform quantization scheme $Q_\eta$ that acts coordinatewise as $Q_\eta(x)=\floor{x/\eta}\eta$. We proved in \Cref{thm:UnifQuantifErrorNecessaryCond} that the number of bits per coordinate must grow linearly with the depth of the network in order to provide $\eps$-error in $L^\infty([-D,D]^d)$, uniformly on a bounded set of parameters $\paramarchiclass{L}{N}^q(r)$. The proof exploits a new lower-bound on the Lipschitz constant of the parameterization of ReLU networks that we established in \Cref{thm:BoundsLipschitzConstant}. We also proved a generic upper-bound for this Lipschitz constant, which generalizes upper-bounds known in specific situations. As a consequence, we gave explicit conditions on the number of bits per coordinate that guarantees quantized ReLU networks to have the same approximation speeds as unquantized ones in generic $L^p$ spaces, see \Cref{ex:CompSpeedQuantizedNN}. We further used in \cref{{sec:UnifQuantifNN}} the upper bound on the Lipschitz constant of $\theta\mapsto R_\theta$ to recover a known approximation result of quantized ReLU networks in $L^\infty$-Sobolev spaces \cite[Thm. 2]{Ding19exprQuantizedNN} and to improve a result on the error of nearest-neighbour uniform quantization \cite[Lem. \MakeUppercase{\romannumeral 6}.8]{Elbrachter21DNNApproximationTheory}.

\textbf{Notion of $\gamma$-encodability.} This paper introduced in \Cref{def:GammaEncodability} a new property of approximation families: being $\speed$-encodable. As soon as $\approxclass$ is $\speed$-encodable in a metric space $(\F,d)$, \Cref{thm:unifying_thm} shows that there is a simple relation between the approximation speed of every set $\fctclass\subset\F$ and its encoding speed:
\begin{equation}\label{ineq:ccl}
    \min (\approximationspeed(\fctclass|\approxclass),\speed )\leq \encodingspeed(\fctclass).
\end{equation}
As seen in \cref{sec: gamma-encodable}, several classical approximation families $\approxclass$ are $\speed$-encodable for some $\speed>0$, including classical families defined with dictionaries (\cref{subsec:InfEncDict}) or ReLU neural networks (\cref{subsec:NNInfEnc}). As a consequence, $\speed$-encodability lays a generic framework that unifies several situations where Inequality \eqref{ineq:ccl} is known, such as when doing approximation with dictionaries \cite[Thm. 5.24]{Grohs15OptimallySparseDataRep}\cite[Prop. 11]{Kerkyacharian04Entropy} or ReLU neural networks \cite[Thm. \MakeUppercase{\romannumeral 6}.4]{Elbrachter21DNNApproximationTheory}.

\textbf{Perspectives. } In \Cref{thm:UnifQuantifErrorSufficientCond} and \Cref{thm:UnifQuantifErrorNecessaryCond}, we saw necessary and sufficient conditions on $\eta>0$ to guarantee that quantizing coordinatewise by $Q_\eta(x)=\floor{x/\eta}\eta$ provides $\eps$-error in $L^\infty([-D,D]^d)$, uniformly on a bounded set of parameters $\paramarchiclass{L}{N}^q(r)$. In practical applications with post-training quantization, we are only interested in parameters that can be obtained with learning algorithms such as stochastic gradient descent. Moreover, we may not be interested in $\eps>0$ arbitrary small. For instance, quantization aware training techniques~\cite{Courbariaux15BinaryConnect} have been successfully applied for ReLU neural networks with three hidden layers and 1024 neurons per hidden layer \cite{Courbariaux15BinaryConnect}. Indeed, the modified learning procedure yields in \cite{Courbariaux15BinaryConnect} a network with quantized weights in $\{-1,1\}$ that performs similarly, on the MNIST dataset, as the network that would have been obtained with the original learning procedure. Is it possible to have better guarantees if we only care about some prescribed error $\eps>0$ and a "small set" of parameters, such as parameters than can indeed be learned in practice?

Another question would be to design schemes to quantize network parameters, in a way that adapts to the architecture. In the quantization schemes covered by \Cref{thm:UnifQuantifErrorSufficientCond}, the sufficient value of $\eta>0$ to ensure a prescribed error $\eps>0$ only takes into account the depth and the width of the architecture. However, in practice the network architecture is carefully designed to meet some criterion, such as reducing the inference cost (references can be found in the paragraph "Compact network design" of \cite{ZhangYYH18LQNets}). Specificities of the architecture could be taken into consideration when designing the quantization scheme.

Another perspective is to take into account functionally equivalent parameters when designing a quantization scheme, as we now detail. Given parameters $\theta$ of a ReLU neural network (and possibly a finite dataset), we say that $\theta'$ is functionally equivalent to $\theta$, denoted $\theta'\sim\theta$, if $R_{\theta}=R_{\theta'}$ (resp. equality on the considered dataset). Due to the positive homogeneity of the ReLU function, there are uncountably many equivalent parameters to $\theta$ that can be obtained by rescaling the coordinates of $\theta$ (but these are not the only ones since permuting coordinates can also lead to functionally equivalent parameters). When quantizing $\theta$, it would be interesting to take these equivalent parameters into account.

Finally, what is the minimum number of bits per coordinate needed to keep the same approximation speeds? While the question remains open, \Cref{thm:ApproxSpeedQuantifNN} makes a first step in that direction by giving an upper-bound.
\bibliographystyle{abbrv}

\appendices

\section{Norms}\label{app: norms}
\noindent
\begin{definition}[$p$-norm]
    Let $d\in\mathbb{N}$. For an exponent $p\in[1,\infty]$, the $p$-norm on $\R^d$ is defined by:
    \begin{align*}
        \forall x=(x_i)_{i=1,\dots,d}\in\mathbb{R}^d, \|x\|_p:= \left\{\begin{array}{ll}
                                                                           \left(\sum_{i=1}^d |x_i|^p\right)^{\frac{1}{p}} & \mbox{ if } p<\infty, \\
                                                                           \sup\limits_{i=1,\dots,d} |x_i|                 & \mbox{ if } p=\infty.
                                                                       \end{array}\right.
    \end{align*}
\end{definition}

\begin{definition}($\vertiii{\cdot}_p$)
    Let $d_1,d_2\in\mathbb{N}$. The operator norm $\vertiii{\cdot}_p$ on $\mathbb{R}^{d_2\times d_1}$ associated with the exponent $p\in[1,\infty]$ is defined by:
    \[
        \forall M\in\mathbb{R}^{d_2\times d_1}, \vertiii{M}_p:=\sup\limits_{\substack{x\in\mathbb{R}^{d_1}\\ x\neq 0}} \frac{\|Mx\|_p}{\|x\|_p}.
    \]
\end{definition}
\section{Characterization of the $L^p$ spaces containing all the functions realized by ReLU networks}\label{app:CharacLpSpaces}
\noindent
\begin{proof}[Proof of of \Cref{lem:CharacLpSpaces}]
    Assume that $C_{p}(\Omega,\mu) < \infty$ and consider the realization $R_{\param}$ of an arbitrary ReLU network on an arbitrary architecture with input dimension $N_0=d_{\textrm{in}}$ and arbitrary output dimension $N_L$. It is known \cite[Thm. 2.1]{arora2018understandingReLUnetworks} that $R_{\param}$ is (continuous and) piecewise linear, so that there is a partition of $\Omega$ into finitely many $\Omega_{i}$, $1 \leq i \leq n$ such that $R_{\param} = \sum_{i=1}^{n} \chi_{\Omega_{i}} f_{i}$ where $\chi_{E}(x)$ is the characteristic function of the set $E$ and each $f_{i}$ is an affine function. To prove the result it is thus sufficient to show that $\chi_{E} f \in L^p(\Omega\to\R^{N_L},\mu)$ for each set $E \subset \Omega$ and each affine function $f$. Since $\|\chi_{E}g\|_{p} \leq \|g\|_{p}$ for any $g$ it is enough to prove that any affine function is in the desired space. For this, consider arbitrary $A \in\R^{N_L \times N_0}, b \in\R^{N_L}$, and $f: x \mapsto Ax+b$. Denoting $c(f) := \max(\vertiii{A}_{\infty},\|b\|_{\infty})$ (the notation $\vertiii{\cdot}$ is defined in appendix \ref{app: norms}) we observe that $\|f(x)\|_{\infty} \leq \vertiii{A}_{\infty}\|x\|_{\infty}+\|b\|_{\infty} \leq c(f) (\|x\|_{\infty}+1)$ so that $\|f\|_{p} \leq c(f) C_{p}(\Omega,\mu) < \infty$, showing the result.

    Conversely, assume that for every architecture $(L,\mathbf{N})$ and parameter $\theta \in \Theta_{L,\mathbf{N}}$ we have $R_{\theta} \in L^p(\Omega\to\R^{N_{L}},\mu)$. Specializing to an architecture with $L=1$, $N_{1}=N_{0}=d_{\textrm{in}}$, consider
    $\theta = (W_{1},b_{1})$ with $W_{1}$ the identity matrix and $b_{1}$ the zero vector,
    $\theta' = (W'_{1},b'_{1})$ with $W'_{1}$ the zero matrix and $b'_{1}$ any vector with $\|b'_{1}\|_{\infty}=1$.
    We have $R_{\theta}(x) = x$ while $R_{\theta'}(x) = b'_{1}$. For $p<\infty$ we have $\int_{x \in \Omega} \|x\|_{\infty}^{p} \mathrm{d}\mu(x)  = \|R_{\theta}\|_{p}^{p} < \infty$ and $\int_{x \in \Omega} 1 \mathrm{d}\mu(x)  = \|R_{\theta'}\|_{p}^{p} < \infty$. By the triangle inequality we get $C_{p}(\Omega,\mu)  < \infty$. The case $p=\infty$ is similar.
\end{proof}
\section{Optimality of a bound on $\|R_\theta(x)-R_{\theta'}(x)\|_q$} \label{app:optimality}
\noindent We generalize a known inequality established for $q=\infty$  \cite[Eq. (37)]{Elbrachter21DNNApproximationTheory}\cite[Eq. (3.12)]{Bolcskei2018optimal} to arbitrary $q$-th norm $q\in[1,\infty]$. Moreover, we prove its optimality. This inequality is used in appendix \ref{app: proof Lips constant} to bound the Lipschitz constant of the parameterization of ReLU networks. With $I_{m \times m}$ the identity matrix in dimension $m$ and $0_{m \times n}$ the $m\times n$ matrix full of zeros, we introduce the following notation for ``rectangular identity matrices'': for $m < n$, we set $I_{m \times n} = (I_{m \times m}; 0_{m \times (n-m)})$, while for $m>n$ we set $I_{m \times n} = I_{n \times m}^{\top}$.
\begin{lemma}\label{lem: ineq norm(rparam-rparam') and optimality}
    Let $\archi{L}{N}$ be an architecture with any depth $L\geq 1$ and $\param=(W_1,\dots,W_L,b_1,\dots,b_L)$, $\param'=(W'_1,\dots,W'_L,b'_1,\dots,b'_L)\in\paramarchiclass{L}{N}$ (see Equation~\eqref{eq: def Theta L N et dim(L,N)} for the definition of $\paramarchiclass{L}{N}$) be parameters associated to this architecture. For every $\ell=1,\dots,L-1$, define $\param'_{\ell}$ as the parameter deduced from $\param'$, associated to the architecture $(\ell,(N_0,\dots, N_\ell))$:
    \[
        \param'_{\ell}=(W'_1,\dots,W'_\ell,b'_1,\dots,b'_\ell).
    \]
    Then for every exponent $q\in[1,\infty]$ and for every $x\in\mathbb{R}^{N_0}$, the realization of neural networks with any $1$-Lipschitz activation function $\varrho$ such that $\varrho(0)=0$ satisfy:
    \begin{align}\label{ineq: rparam-rparam'}
        \|R_{\param}(x)-R_{\param'}(x)\|_q \leqslant {} & \sum_{\ell=1}^L\left(\prod_{k=\ell+1}^L\vertiii{W_k}_q\right)\times\vertiii{W_\ell-W_\ell'}_q\times\|R_{{\param}'_{\ell-1}}(x)\|_q \\
                                                        & + \sum_{\ell=1}^L\left(\prod_{k=\ell+1}^L\vertiii{W_k}_q\right)\|b_\ell-b'_\ell\|_q,\nonumber
    \end{align}
    where the definition of the $q$-th norm and the operator norm of a matrix are recalled in appendix \ref{app: norms}, and where we set by convention $R_{{\param}'_{\ell-1}}(x)=x$ if $\ell=1$, and $\prod_{k=\ell+1}^L\vertiii{W_k}_q = 1$ if $\ell=L$.

    Let $\lambda_1,\dots,\lambda_L\geqslant 0$ and $\eps\geqslant 0$ and consider an input vector $x \in \mathbb{R}^{d_{\textrm{in}}}$ with nonnegative entries and supported on the first $s := \min_{\ell} N_{\ell}$ coordinates. There is equality in~\eqref{ineq: rparam-rparam'} for the parameters $\param=(W_1,\dots,W_L,b_1,\dots,b_L)$ and $\param'=(W'_1,\dots,W'_L,b'_1,\dots,b'_L)$ defined by, for every $\ell=1,\dots, L$:
    \begin{equation}\label{eq:ParamEqCase}
        W_\ell = \lambda_\ell I_{N_{\ell} \times N_{\ell-1}},\quad
        W'_\ell = (1+\eps)W_\ell,\quad b_\ell = b'_\ell = 0.
    \end{equation}
\end{lemma}

\begin{proof} The proof of Inequality \eqref{ineq: rparam-rparam'} follows by induction on $L\in\N$ in a similar way as in the case $q=\infty$ \cite[Eq. (37)]{Elbrachter21DNNApproximationTheory}\cite[Eq. (3.12)]{Bolcskei2018optimal}. For $L=1$, this is just saying that
    \begin{align*}
        \|R_{\param}(x)-R_{\param'}(x)\|_q {} & = \|W_1x+b_1 - W_1'x - b'_1\|_q                    \\
        {}                                    & \leq \vertiii{W_1-W'_1}_q\|x\|_q + \|b_1-b'_1\|_q.
    \end{align*}
    Assume that the property holds true for $L\geqslant 1$. Then at rank $L+1$ (using in the last inequality that the activation function $\rho$ is $1$-Lipschitz and $\varrho(0)=0$):
    \begin{align*}
        \|R_{\param}(x)-R_{\param'}(x)\|_q = {} & \|W_{L+1}\rho(R_{\param_L}(x))+b_{L+1} - W'_{L+1}\rho(R_{\param'_L}(x))-b'_{L+1}\|_q \\
        = {}                                    & \| W_{L+1}\left(\rho(R_{\param_L}(x))-\rho(R_{\param'_L}(x))\right)                  \\
        {}                                      & + (W_{L+1}-W'_{L+1})\rho(R_{\param'_L}(x)) + b_{L+1}-b'_{L+1}\|_q                    \\
        \leq {}                                 & \vertiii{W_{L+1}}_q\|\rho(R_{\param_L}(x)) - \rho(R_{\param'_L}(x))\|_q              \\
        {}                                      & + \vertiii{W_{L+1}-W'_{L+1}}_q\|\rho(R_{\param'_L}(x))\|_q+ \|b_{L+1}-b'_{L+1}\|_q   \\
        \leq {}                                 & \vertiii{W_{L+1}}_q\|R_{\param_L}(x) - R_{\param'_L}(x)\|_q                          \\
        {}                                      & + \vertiii{W_{L+1}-W'_{L+1}}_q\|R_{\param'_L}(x)\|_q+ \|b_{L+1}-b'_{L+1}\|_q.
    \end{align*}
    Using the induction hypothesis gives the desired result.

    For the equality case, recall the definition of the parameters $\theta$ and $\theta'$ in \Cref{eq:ParamEqCase}. Let $\lambda = \prod_{\ell=1}^L \lambda_\ell$. Since $x = (y^\top,0_{1 \times (d_{\textrm{in}}-s)})^{\top}$ with $y \in \mathbb{R}_{+}^{s}$ we have $\varrho(W_{1}x+b_{1}) = \lambda_1 (y^{\top},0_{1 \times (N_{1}-s)})^{\top}$. By induction on $\ell=1,\dots,L$, we can show $R_\theta(x)=\lambda (y^{\top},0_{1 \times (N_{L}-s)})^{\top}$, and similarly $R_{\theta'}(x)=(1+\eps)^L\lambda (y^{\top},0_{1 \times (N_{L}-s)})^{\top}$. This means that:
    \begin{align*}
        \|R_{\param}(x)-R_{\param'}(x)\|_q {} & = \|\lambda y - (1+\eps)^L\lambda y\|_q \\
        {}                                    & = ((1+\eps)^L - 1)\lambda\|x\|_q.
    \end{align*}
    Moreover, for every $\ell=1,\dots, L$, it is easy to check that $\vertiii{W_\ell}_q=\lambda_\ell$, $\vertiii{W'_\ell}_q=(1+\eps)\lambda_\ell$ and $\vertiii{W_\ell-W'_\ell}_q=\eps\lambda_\ell$ so that:
    \begin{align*}
        \left(\prod_{k=\ell+1}^L\vertiii{W_k}_q\right) {} & \times\vertiii{W_\ell-W_\ell'}_q\times\|R_{{\param}'_{\ell-1}}(x)\|_q                                                            \\
        {}                                                & = \left(\prod_{k=\ell+1}^L\lambda_k\right)\times\eps\lambda_\ell\times \left(\prod_{k=1}^{\ell-1}(1+\eps)\lambda_k\right)\|x\|_q \\
        {}                                                & = (1+\eps)^{\ell-1}\eps \lambda\|x\|_q,
    \end{align*}
    and:
    \begin{align*}
        \left(\prod_{k=\ell+1}^L\vertiii{W_\ell}_q\right)\|b_\ell-b'_\ell\|_q = 0.
    \end{align*}
    This yields the equality case, since:
    \begin{multline*}
        \sum_{\ell=1}^L\left(\prod_{k=\ell+1}^L\vertiii{W_k}_q\right) \times\vertiii{W_\ell-W_\ell'}_q\times\|R_{{\param}'_{\ell-1}}(x)\|_q {}
        + \sum_{\ell=1}^L\left(\prod_{k=\ell+1}^L\vertiii{W_\ell}_q\right)\|b_\ell-b'_\ell\|_q {}  \\ = \sum_{\ell=1}^L (1+\eps)^{\ell-1}\eps \lambda\|x\|_q
        {}  = \frac{(1+\eps)^L - 1}{1+\eps-1}\eps\lambda\|x\|_q
        {}  = ((1+\eps)^L - 1)\lambda\|x\|_q.
    \end{multline*}
\end{proof}

\section{Lipschitz parameterization of ReLU networks (Proof of \Cref{thm:BoundsLipschitzConstant})}\label{app: proof Lips constant}
\noindent Recall that we fixed a set $L^p(\Omega\to\R^{d_{\textrm{out}}},\mu)$ containing all functions realized by ReLU neural networks with input dimension $d_{\textrm{in}}$ and output dimension $d_{\textrm{out}}$. The parameter set $\paramarchiclass{L}{N}^q(r)$ is defined in \Cref{def: classe de param associee a une archi et G}.

First, \Cref{lem: ineq norm(rparam-rparam') and optimality} applied to any $\theta \in \paramarchiclass{L}{N}$, and $\theta'=(0,\ldots,0) \in \paramarchiclass{L}{N}$ yields for every $x \in \Omega$:
\begin{align}\label{eq:TmpBoundRealiz}
    \|R_{\param}(x)\|_q {} & \leqslant \prod_{k=1}^{L}\vertiii{W_k}_q\|x\|_q + \sum_{\ell=1}^{L}\left(\prod_{k=\ell+1}^{L}\vertiii{W_k}_q\right)\|b_\ell\|_q,
\end{align}
using that $\|R_{{\param}'_{\ell-1}}(x)\|_q = \|x\|_{q}$ for $\ell=1$ (by convention) and $\|R_{{\param}'_{\ell-1}}(x)\|_q=0$ for each $\ell \geq 2$ (since $\theta' = 0$).

Let $\param, \param'\in\paramarchiclass{L}{N}$. We are going to bound $\|R_{\param}-R_{\param'}\|_{p,\|\cdot\|}$ from above using Inequality \eqref{ineq: rparam-rparam'} of \Cref{lem: ineq norm(rparam-rparam') and optimality}. First, we introduce useful notations to write things compactly. Define for every $i,j\in\mathbb{N}$:
\begin{align*}
    \Pi_{i,j} {} & := \prod_{k=i}^j \vertiii{W_k}_q\mbox{ and }\Pi'_{i,j}: = \prod_{k=i}^j \vertiii{W'_k}_q \mbox{ if } i\leqslant j, \\
    \Pi_{i,j} {} & := \Pi'_{i,j} := 1 \mbox{ otherwise}.
\end{align*}

For $\ell=2,\dots, L$, we start by bounding $\|R_{\param'_{\ell-1}}(x)\|_{q}$ by a simple function of $x\in\Omega$, since this term appears on the right-handside of Inequality \eqref{ineq: rparam-rparam'}. Using~\eqref{eq:TmpBoundRealiz} for the architecture $(\ell-1,(N_{0},\ldots,N_{\ell-1}))$ we have:
\begin{align*}
    \|R_{\param'_{\ell-1}}(x)\|_q {} & \leqslant \prod_{k=1}^{\ell-1}\vertiii{W'_k}_q\|x\|_q + \sum_{k=1}^{\ell-1}\left(\prod_{j=k+1}^{\ell-1}\vertiii{W'_j}_q\right)\|b'_k\|_q \\
    {}                               & = \Pi'_{1,\ell-1}\|x\|_q + \sum_{k=1}^{\ell-1}\Pi'_{k+1,\ell-1}\|b'_k\|_q.
\end{align*}
If $\Omega \subseteq \mathbb{R}_{+}^{d_{\textrm{in}}}$ and $N_{0} = \min_{0 \leq \ell \leq L} N_{\ell}$ then for every $x \in \Omega$, the parameters defined in \Cref{eq:ParamEqCase} are such that the previous inequality is an equality.

Denote $c_{0}$ a constant such that for every $y\in\R^{d_\textrm{out}}$, $\|y\|\leq c_0\|y\|_q$. Note that if $\|\cdot\|=\|\cdot\|_s$ for $s\in[1,\infty]$, then we can take $c_{0}=d_{\textrm{out}}^{\max\left(0,\frac{1}{s}-\frac{1}{q}\right)}$. Now, using the previous inequality and integrating both sides of Inequality \eqref{ineq: rparam-rparam'} of \Cref{lem: ineq norm(rparam-rparam') and optimality}, we get for $1 \leq p < \infty$:
\begin{multline*}
    \biggl(\int_{x\in\Omega} \|R_{\param}(x) - R_{\param'}(x)\|^p \mathrm{d}\mu(x)\biggr)^{\frac{1}{p}} {}  \leq c_{0}
    \biggl(\int_{x\in\Omega}\|R_{\param}(x) - R_{\param'}(x)\|_q^p \mathrm{d}\mu(x)\biggr)^{\frac{1}{p}} \\
    {}  \leq c_{0}\biggl(\int_{x\in\Omega}\biggl[\sum_{\ell=1}^L \Pi_{\ell+1,L}\left(\Pi'_{1,\ell-1}\|x\|_q + \sum_{k=1}^{\ell-1}\Pi'_{k+1,\ell-1}\|b'_k\|_q\right)\\
        {}  \times \vertiii{W_\ell-W'_\ell}_q + \sum_{\ell=1}^L\Pi_{\ell+1,L}\times \|b_\ell-b'_\ell\|_q\biggr]^p\mathrm{d}\mu(x)\biggr)^{\frac{1}{p}}.
\end{multline*}
A trivial adaptation yields a similar result for $p = \infty$.

If $\Omega \subseteq \mathbb{R}_{+}^{d_{\textrm{in}}}$, $N_{0} = \min_{0 \leq \ell \leq L} N_{\ell}$, and if $\|\cdot\|=\|\cdot\|_q$ so that we can take $c_0:=1$, then the previous inequality is an equality for the parameters defined in Equation~\eqref{eq:ParamEqCase}.

Note that in the special case $p=\infty$, if we only assume that $\Omega \subseteq \mathbb{R}_{+}^{d_{\textrm{in}}}$ and $\|\cdot\|=\|\cdot\|_q$ (but not that $N_{0} = \min_{0 \leq \ell \leq L} N_{\ell}$), denoting by $N_{\min}:= \min_{0 \leq \ell \leq L} N_{\ell}$, then it holds for the parameters of Equation~\eqref{eq:ParamEqCase} and for every $x\in\Omega$ supported on the first $N_{\min}$ coordinates:
\[
    \|R_{\theta}(x)-R_{\theta'}(x)\| =  \sum_{\ell=1}^L \Pi_{\ell+1,L}\left(\Pi'_{1,\ell-1}\|x\|_q + \sum_{k=1}^{\ell-1}\Pi'_{k+1,\ell-1}\|b'_k\|_q\right) \times \vertiii{W_\ell-W'_\ell}_q + \sum_{\ell=1}^L\Pi_{\ell+1,L}\times \|b_\ell-b'_\ell\|_q.
\]

Recall that $W=\max\limits_{\ell=0,\dots L} N_\ell$ is the width of the network. For every matrix $M$ with input/output dimension bounded by $W$ and every vector $b$ with dimension bounded by $W$, denoting by $\|M\|_{\max}:=\max_{i,j} |M_{i,j}|$, standard results on equivalence of norms guarantees that for every $1\leq q\leq \infty$, it holds $\|b\|_q\leq W^{1/q}\|b\|_\infty\leq W\|b\|_\infty$ and $\max(\vertiii{M}_1,\vertiii{M}_{\infty})\leq W\|M\|_{\max}$. The latter, with Riesz-Thorin theorem \cite[Chap.2, Thm 4.3]{DeVore1993}, guarantee that for every $1\leq q\leq \infty$:
\begin{equation}\label{eq:OperatorNormEquivalence}
    \vertiii{M}_q\leqslant W\|M\|_{\max} \mbox{ and } \|b\|_q\leqslant W\|b\|_{\infty}.
\end{equation}
We deduce that for every $\ell=1,\dots, L$:
\begin{align*}
    \max\left(   \vertiii{W_\ell-W'_\ell}_q, \| b_\ell-b'_\ell\|_q\right) \leqslant W\|\param-\param'\|_{\infty}.
\end{align*}
This time, this is not an equality for the parameters defined in \Cref{eq:ParamEqCase}. For them it holds instead, assuming that all $\lambda_{\ell}$ are equal:
\[
    \vertiii{W_\ell-W'_\ell}_q =\eps\lambda_{\ell} =  \|W_\ell-W'_\ell\|_{\max} =  \|\param-\param'\|_{\infty}, \|b_\ell-b'_\ell\|_q = 0.
\]
Using the previous inequalities, we get for $1 \leq p < \infty$:
\begin{align*}
    \|R_{\param} - R_{\param'}\|_{p,\|\cdot\|}
    \leqslant {} & \biggl(\int_{x\in\Omega}\biggl[\sum_{\ell=1}^L \Pi_{\ell+1,L}\left(\Pi'_{1,\ell-1}\|x\|_q + \sum_{k=1}^{\ell-1}\Pi'_{k+1,\ell-1}\|b'_k\|_q\right) \\
    {}           & + \sum_{\ell=1}^L\Pi_{\ell+1,L}\biggr]^p\mathrm{d}\mu(x)\biggr)^{\frac{1}{p}} c_{0} W\|\param-\param'\|_{\infty}
\end{align*}
with a trivial adaptation for $p=\infty$.
Now, let's specialize this for $\param, \param'\in\paramarchiclass{L}{N}^q(r)$. It holds $\max(\Pi_{i,j}, \Pi'_{i,j})\leq r^{j-i+1}$ for $i \leq j$, and the same also holds for $i = j+1$ by definition of $\Pi_{i,j}$. Thus:
\begin{align*}
    {} & \sum_{\ell=1}^L\Pi_{\ell+1,L}\left(1+\Pi'_{1,\ell-1}\|x\|_q + \sum_{k=1}^{\ell-1}\Pi'_{k+1,\ell-1}\|b'_k\|_q\right)                                             \\
    {} & \leqslant \sum_{\ell=1}^Lr^{L-\ell}\left(1+r^{\ell-1}\|x\|_q + \sum_{k=1}^{\ell-1}r^{\ell-k}\right) \mbox{ since } \param,\param'\in\paramarchiclass{L}{N}^q(r) \\
    {} & = Lr^{L-1}\|x\|_q+ \sum_{\ell=1}^Lr^{L-\ell} + \sum_{\ell=1}^L\sum_{k=1}^{\ell-1}r^{L-k}                                                                        \\
    {} & \leqslant Lr^{L-1}\|x\|_q + Lr^{L-1} + L(L-1)r^{L-1} \mbox{ since } r\geqslant 1                                                                                \\
    {} & \leqslant L^2r^{L-1}(\|x\|_q + 1) \mbox{ since } L\geqslant 1.
\end{align*}
If we define:
\[
    c :=
    \begin{cases}
        c_{0} \left(\int_{x\in\Omega} (\|x\|_{q}+1)^p\mathrm{d}\mu(x)\right)^{1/p} & \mbox{ if } p<\infty, \\
        c_{0} \esssup\limits_{x\in\Omega} \|x\|_{q} + 1                            & \mbox{ if } p=\infty.
    \end{cases}
\]
where we recognize in the second factor the constant $C_{p}(\Omega,\mu)$ from \Cref{lem:CharacLpSpaces} when $q=\infty$, then we finally get \eqref{eq:LipsConstDNN}.
Let us now explicit $c$ in specific situations where
$\Omega=[-D,D]^d$ for some $D>0$, $\mu$ is the Lebesgue measure and $\|\cdot\|=\|\cdot\|_q$ so that we can take $c_0=1$. If $q = \infty$ we get $c=C_{p}(\Omega,\mu) \leq (D+1)(2D)^{d/p}$.
If $p=\infty$, then
\[
    c= \esssup\limits_{x\in\Omega} \|x\|_{q} + 1 = Dd^{1/q}+1.
\]
Indeed, the essential supremum is actually a maximum in this case and $\|x\|_q\leq d^{1/q}\|x\|_\infty\leq d^{1/q}D$ for every $x\in[-D,D]^d$ with equality for $x=(D,\ldots, D)^T$.

Let us now discuss the optimality of \eqref{eq:LipsConstDNN}. It can be checked that if $\Omega \subseteq \mathbb{R}_{+}^{d_{\textrm{in}}}$, $\|\cdot\|=\|\cdot\|_q$, so that we can take $c_0:=1$, and if $N_{0} = \min_{0 \leq \ell \leq L} N_{\ell}$, then the parameters $\param,\param'$ defined in \Cref{eq:ParamEqCase} with $\lambda_1=\ldots=\lambda_L=\frac{r}{1+\eps}\geq 0$ are in $\paramarchiclass{L}{N}^q(r)$ and satisfy $\|R_{\param}-R_{\param'}\|_p = c_0\left(\int_{x\in\Omega} \|x\|_q^p\mathrm{d}\mu(x))\right)^{\frac{1}{p}}r^{L-1}\sum_{\ell=1}^L \left(\frac{1}{1+\eps}\right)^{L-\ell}\|\theta-\theta'\|_\infty$.

In the special case where $p=\infty$, if we only assume that $\Omega \subseteq \mathbb{R}_{+}^{d_{\textrm{in}}}$ and $\|\cdot\|=\|\cdot\|_q$ then the parameters $\param,\param'$ defined in \Cref{eq:ParamEqCase} with $\lambda_1=\ldots=\lambda_L=\frac{r}{1+\eps}\geq 0$ are in $\paramarchiclass{L}{N}^q(r)$ and if we denote $N_{\min}:= \min_{0 \leq \ell \leq L} N_{\ell}$ and $\Omega_{\min}$ the set of $x\in\Omega$ supported on the first $N_{\min}$ coordinates:
\[
    \esssup\limits_{x\in\Omega_{\min}} \|R_\theta(x)-R_{\theta'}(x)\|\geq \left(\esssup\limits_{x\in\Omega_{\min}} \|x\|_q\right)r^{L-1}\sum_{\ell=1}^L \left(\frac{1}{1+\eps}\right)^{L-\ell}\|\theta-\theta'\|_\infty.
\]
This yields the conclusion.
\section{Nearest-neighbour uniform quantization on ReLU networks}\label{app:CharacUnifQuantif}
%
\begin{proof}[Proof of \Cref{thm:UnifQuantifErrorNecessaryCond}]
    Consider $\eps,\eta>0$ such that \eqref{eq:EtaUnifQuantEpsDist} holds true. We must prove that $\min(r,\eta)\leq \frac{\eps}{c'r^{L-1}}$. With $I_{m \times m}$ the identity matrix in dimension $m$ and $0_{m \times n}$ the $m\times n$ matrix full of zeros, we introduce the following notation for ``rectangular identity matrices'': for $m < n$, we set $I_{m \times n} = (I_{m \times m}; 0_{m \times (n-m)})$, while for $m>n$ we set $I_{m \times n} = I_{n \times m}^{\top}$. Consider $0 < a <\eta$ and define $\param=(W_1,\dots, W_L, b_1,\dots, b_L)$ with $b_1=\dots=b_L=0$, $W_1 = \lambda I_{N_{1} \times N_{0}}$ with $\lambda :=  \min(r,(\eta-a))$, and for every layer $\ell\geq 2$, $W_\ell= r I_{N_{\ell} \times N_{\ell-1}}$. Since $0<\lambda\leq \eta-a < \eta$, we have $Q_\eta(\lambda)= 0$ so that $Q_\eta(W_1)=0$. Since $b_1=0$, we also have $Q_\eta(b_1)=0$ so that $R_{Q_\eta(\theta)}=0$. We deduce that for every $x\in[0,D]^d$ supported in the first  $N_{\min}$ coordinates:
    \[
        \|R_{\theta}(x)-R_{Q_\eta(\param)}(x)\|_q
        = \|\lambda r^{L-1}x - 0\|_{q} = \lambda r^{L-1}\|x\|_{q}.
    \]
    Since the maximum of $\|x\|_q$ over all $x \in [0,D]^{d}$ supported in the first $N_{\min}$ coordinates is $c' = D N_{\min}^{1/q}$, we get:
    \[
        c' \lambda r^{L-1}\leq
        \max\limits_{x\in[-D,D]^d} \|R_{\theta}(x)-R_{Q_{\eta}(\theta)}(x)\|_q
    \]
    As $\vertiii{W_1}_q = \lambda =  \min(r,(\eta-a))\leq r$, for every $\ell\geq 2$, $\vertiii{W_\ell}_q = r$ and for every $\ell\geq 1$, $\|b_\ell\|_q=0\leq r$, we have $\theta\in\paramarchiclass{L}{N}^q(r)$ so \eqref{eq:EtaUnifQuantEpsDist} applies. This implies $c' \lambda r^{L-1}\leq \eps$, \ie $\min(r,(\eta-a))\leq \eps/(c' r^{L-1})$. This holds for every $0<a<\eta$: taking the limit $a\to 0^+$ yields the result.
\end{proof}
\begin{proof}[Proof of \Cref{prop:Lem6.8}]
    Let us see that \Cref{lem:ErrorQuantifEta} applies with $\Theta=\paramarchiclass{L}{N}^{\max}(\eps^{-k})$, $q=\infty$ and $r = W\eps^{-k}$. First, it holds $\paramarchiclass{L}{N}^{\max}(\eps^{-k})\subset\paramarchiclass{L}{N}^{\infty}(W\eps^{-k})$ (see \Cref{rmk: parameter set with Frobenius or max norm}). Since $\|\theta-\theta'\|_\infty\leq\eta/2 \leq\eps^{m}/2$, the $\eta$ to use for \Cref{lem:ErrorQuantifEta} is $\eta:=\eps^m /2>0$. \Cref{lem:ErrorQuantifEta} then gives the result if $\eps^m /2\leq \eps\left(cWL^2r^{L-1}\right)^{-1}$, where $c:=1+Dd_{\textrm{in}}^{1/q} = 1+D$. The latter condition is equivalent to $cWL^2r^{L-1}/2\leq \eps^{m-1}$. Recall that $r=W\eps^{-k}$ and $\max(W,L)\leq\eps^{-k}$. Thus the left hand-side satisfies: $cWL^2r^{L-1}/2 \leq ((1 +D)/2)\eps^{-2kL-2k}$. We can conclude if $((1 + D)/2)\eps^{-2kL-2k}\leq \eps^{m-1}$. By definition of $m$, this is true as soon as $((1+D)/2)\eps^{\log_2(\ceil{D})} \leq 1$. This is clear when $0< D \leq 1$. While for $D>1$, $((1+D)/2)\eps^{\log_2(\ceil{D})} \leq 1$ holds if and only if $\log_2(\eps)\leq -\frac{\log_2((1+D)/2)}{\log_2(\ceil{D})}$. Since $1<D$, it holds $\frac{1+D}{2}\leq \frac{\ceil{D}+\ceil{D}}{2} = \ceil{D}$ so that $-\frac{\log_2((1+D)/2)}{\log_2(\ceil{D})}\geq -1$. Since $\eps\in(0,1/2)$, it holds $-1\geq \log_2(\eps)$, hence $-\frac{\log_2((1+D)/2)}{\log_2(\ceil{D})}\geq \log_2(\eps)$ and the result follows.
\end{proof}
\begin{proof}[Proof of \Cref{prop:Ding}]
    Using \cite[Thm. 1]{Yarotsky17approxUnitBallSobolevWithDNN}, there exist constants $c(n,d)>0$ and $r(n,d)>1$ (for instance, a proof examination of \cite[Thm. 1]{Yarotsky17approxUnitBallSobolevWithDNN} shows that we can take $r=\max(4,d+n)$) such that for every $\eps \in (0,1)$, there exists a ReLU network architecture $\archi{L}{N}$ with depth $L$ bounded by $c\ln(1/\eps)$, a number of weights at most equal to $c\eps^{-d/n}\ln(1/\eps)$, and such that for every $f\in\fctclass_{n,d}$, there exists $\theta\in\paramarchiclass{L}{N}$ such that $\|f-R_{\theta}\|_{L^\infty([0,1]^d)}\leq \eps/2$, and such that $\theta$ has weight's magnitude bounded by $r$. \Cref{thm:BoundsLipschitzConstant} can now be used to quantize the weights of $\theta$, in order to get a quantized ReLU network $\eps$-close to $f$. Denote $W$ the width of this network architecture $\archi{L}{N}$. Since $\paramarchiclass{L}{N}^{\max}(r)\subset\paramarchiclass{L}{N}^1(Wr)$ (see \Cref{rmk: parameter set with Frobenius or max norm}) we can use \Cref{thm:BoundsLipschitzConstant} with $q=1$ to get that there exists a constant $c'>0$ that only depends on $n,d$, such that the weights of any network $\theta \in \paramarchiclass{L}{N}^{\max}(r)$ can be uniformly quantized with a step size $\eta:=c'\eps(WL^{2}(Wr)^{L-1})^{-1}$ to get a quantized network $\theta'$ such that $\|R_{\theta'}-R_{\theta}\|_{L^{\infty}([0,1]^{d})} \leq \eps/2$. Since the width $W$ is at most the number of weights, which is at most $c\eps^{-d/n}\ln(1/\eps)$, and since the depth $L$ is at most $c\ln(1/\eps)$ and $r$ is a constant that only depends on $n,d$, it is straightforward to check that $\ln(1/\eta)\leq c''\ln^2(1/\eps)$ for some constant $c''$ that only depends on $n$ and $d$. Since the weights are bounded in absolute value by $r(n,d)$, this means that every quantized weight can be stored using at most $c'''\ln(1/\eta)\leq c'''\ln^2(1/\eps)$ bits for some constant $c'''(n,d)>0$. Since there are at most $c\eps^{-d/n}\ln(1/\eps)$ such quantized weights, this yields the result using $\max(c,c'',c\times c''')$
    as the final constant.
\end{proof}
\section{Approximation speeds of quantized versus unquantized ReLU networks}\label{app:LipsParamUnifQuant}
\noindent We first establish two lemmas that will be useful to prove \Cref{thm:ApproxSpeedQuantifNN}. Along the way, we also give bounds on the size of the coverings we encounter. These bounds will prove useful in appendix \ref{app:DictNNInfEnc}.
\begin{lemma}\label{lem:ApproxSpeedGammaEncod}
    Let $(\F,d)$ be a metric space. Consider $\speed$ and two sequences $\approxclass(\speed)$ and $\approxclass$ of subsets of $\F$. Assume that there exists a constant $c>0$ such that for every $M\in\N$, the set $\approxclass(\speed)_M$ is a $cM^{-\gamma}$-covering of $\approxclass_M$. Then for every (non-empty) $\fctclass\subset\F$:
    \begin{equation*}
        \begin{array}{ll}
            \approximationspeed(\fctclass|\approxclass(\speed)) = \approximationspeed(\fctclass|\approxclass) & \text{if $\speed\geq \approximationspeed(\fctclass|\approxclass)$}, \\
            \approximationspeed(\fctclass|\approxclass(\speed)) \geq \speed                                   & \text{otherwise}.
        \end{array}
    \end{equation*}
\end{lemma}
\begin{proof}[Proof of \Cref{lem:ApproxSpeedGammaEncod}]
    For every $M\in\N$, the inclusion $\approxclass(\speed)_M\subset\approxclass_M$ holds (indeed $\approxclass(\speed)_M$ is a covering of $\approxclass_M$) so that $\approximationspeed(\fctclass|\approxclass(\speed))\leq \approximationspeed(\fctclass|\approxclass)$. This proves the result when $\approximationspeed(\fctclass|\approxclass) = -\infty$. From now on we assume $\approximationspeed(\fctclass|\approxclass)>-\infty$. Fix an arbitrary
    $-\infty < \speed'<\min(\approximationspeed(\fctclass|\approxclass),\speed)$. By definition of the approximation speed, there exists a constant $c'>0$ such that for every $f\in\fctclass$ and every $M\in\N$, there exists a function $\Phi_M(f)\in\approxclass_M$ that satisfies:
    \[
        d\left(f,\Phi_M(f)\right) \leq c'M^{-\speed'}.
    \]
    The triangle inequality guarantees that for every $f\in\fctclass$ and every $M\in\N$:
    \[
        d(f,\approxclass(\speed)_M) \leq d(f,\Phi_M(f)) + d(\Phi_M(f),\approxclass(\speed)_M) \leq c'M^{-\speed'}+cM^{-\speed}.
    \]
    Since $\gamma' \leq \gamma$ (and even if $\gamma' < 0$, which can happen if $\approximationspeed(\fctclass|\approxclass)<0$) this means that $\approximationspeed(\fctclass|\approxclass(\speed))\geq \speed'$ for every $-\infty<\speed'<\min(\approximationspeed(\fctclass|\approxclass),\speed)$ so $\approximationspeed(\fctclass|\approxclass(\speed)) \geq \min\left(\approximationspeed(\fctclass|\approxclass),\speed\right)$. Since we also proved that $\approximationspeed(\fctclass|\approxclass(\speed))\leq \approximationspeed(\fctclass|\approxclass)$, this yields the claim.
\end{proof}
\begin{lemma}\label{lem:CovNumberLipImage}
    Consider $q\in[1,\infty]$ and $\speed>0$. There exists a constant $c(q,\speed)>0$ such that the following holds. Consider arbitrary $n\in\N$, $r\geq 1$ and consider the set $B_{n,\|\cdot\|_q}(0,r)\subset\ell^q(\N)$ that consists of the sequences bounded in $\ell^q$-norm by $r$, and with zero coordinates outside the first $n$ ones. Consider a metric space $(\F,d)$ and a Lipschitz-map $\varphi:(B_{n,\|\cdot\|_q}(0,r),\|\cdot\|_q)\to(\F,d)$ with Lipschitz constant $\Lips(\varphi)\geq 1$. For every $M\in\N$, define the step size $\eta_M:=(M^{\gamma}n^{1/q}\Lips(\varphi))^{-1}$ and the "quantized" set $\mQ(B_{n,\|\cdot\|_q}(0,r),\eta_M,\infty) := B_{n,\|\cdot\|_q}(0,r)\cap (\eta_M\Z)^{\N}$. Then for every integer $M\geq 2$, the set $\varphi(\mQ(B_{n,\|\cdot\|_q}(0,r),\eta_M,\infty))$ is an $M^{-\gamma}$-covering of $\varphi(B_{n,\|\cdot\|_q}(0,r))$ of size satisfying:
    \begin{equation}\label{eq: size covering}
        \log_2(|\varphi(\mQ(B_{n,\|\cdot\|_q}(0,r),\eta_M,\infty))|) \leq c(q,\speed)\biggl(n\biggl[\log_2(n) + \log_2(r)+\log_2(\Lips(\varphi))+\log_2(M)\biggr]\biggr).
    \end{equation}
\end{lemma}
\begin{proof}
    When $q=\infty$, it is known \cite[Examples 5.2 and 5.6]{Wainwright19HDS} that $\mQ(B_{n,\|\cdot\|_q}(0,r),\eta_M,\infty)$ is a $\eta_M$-covering of $B_{n,\|\cdot\|_q}(0,r)$ of size bounded by $(2r/\eta_M)^{n}+1$. Since $\varphi$ is $\Lips(\varphi)$-Lipschitz, we deduce that the set $\varphi(\mQ(B_{n,\|\cdot\|_q}(0,r),\eta_M,\infty))$ is an $M^{-\gamma}$-covering of $\varphi(B_{n,\|\cdot\|_q}(0,r))$ of size satisfying:
    \[
        \log_2(|\varphi(\mQ(B_{n,\|\cdot\|_q}(0,r),\eta_M,\infty))|) \leq n\biggl[1+ \log_2(r)+\log_2(\Lips(\varphi))+\speed\log_2(M)\biggr]
    \]
    Since $M\geq 2$, it holds $1+\speed\log_2(M)\leq (1+\speed)\log_2(M)$, hence Equation~\eqref{eq: size covering} for $c(q,\speed)=1+\speed\geq 1$. This settles the case $q=\infty$.

    When $q\in[1,\infty)$, Hölder's inequality yields $\|x\|_q\leq n^{1/q}\|x\|_\infty$ for every $x\in\R^n$. Thus $B_{n,\|\cdot\|_q}(0,r)$ is a subset of the ball of radius $rn^{1/q}$ of $\ell^{\infty}(\N)$, and the Lipschitz constant of $\varphi$ with respect to $\|\cdot\|_{\infty}$ is bounded by its Lipschitz constant with respect to $\|\cdot\|_{q}$, up to a factor $n^{1/q}$. Thus, the case $q\in[1,\infty)$ can be reduced to the case $q=\infty$ by replacing $r$ by $rn^{1/q}$ and $\Lips(\varphi)$ by $n^{1/q}\Lips(\varphi)$. We get:
    \[
        \log_2(|\varphi(\mQ(B_{n,\|\cdot\|_q}(0,r),\eta_M,\infty))|) \leq n\biggl[1+ \frac{2}{q}\log_2(n) + \log_2(r)+\log_2(\Lips(\varphi))+\speed\log_2(M)\biggr]
    \]
    This yields the desired result with $c(q,\speed)=\max(\frac{2}{q},1+\speed)$.
\end{proof}
\begin{lemma}\label{lem:gammaencodNN}
    In the setting of \Cref{thm:ApproxSpeedQuantifNN}, for every $\speed>0$, there exists a constant $c>0$ such that for every $M\in\N$, the set $\mQ(\nnclass|\speed)_M$ is a $cM^{-\gamma}$-covering of $\nnclass$. Moreover, for every $M\in\N$, every architecture $\archi{L}{N}\in\archiclass{M}$ and every support $S\in S^M_{\archi{L}{N}}$:
    \[
        \log_2(|R_{\mQ(\paramarchiclass{L}{N}^q(r_M),\eta_M, r_M),S}|) \leq c M\left(\log_2 (M) + \log_2(r_M) + \log_2(\Lips(M,q))\right).
    \]
\end{lemma}
\begin{proof}
    According to \Cref{thm:BoundsLipschitzConstant}, there is a constant $c'>0$ such that for each $M\in\N$, each architecture $\archi{L}{N}\in\archiclass{M}$ and each support $S\in S^M_{\archi{L}{N}}$, the set $R_{\paramarchiclass{L}{N}^q(r_M),S}$ is the image under a Lipschitz map of $(\{\param\in \paramarchiclass{L}{N}^q(r_M) \text{ supported on }S\}, \|\cdot\|_\infty)$ with a Lipschitz constant bounded by $c'\Lips(M,q)$.

    Let us now use \Cref{lem:CovNumberLipImage} with a Lipschitz constant bounded by $c'\Lips(M,q)$, $n=|S|\leq M$ the cardinality of the support, $r=r_M$, and the same $q$ as here. This yields the result with $c:=\max(1/c', 2c(q,\speed))$.
\end{proof}
\begin{proof}[Proof of \Cref{thm:ApproxSpeedQuantifNN}]
    Combining \Cref{lem:ApproxSpeedGammaEncod} and \Cref{lem:gammaencodNN} gives Equality \eqref{eq:QuantSpeedNN}.
\end{proof}
\section{Encodability implies a relation between approximation and encoding speeds}\label{app:EncImpliesUnifyingBound}
\begin{proof}[Proof of \Cref{thm:unifying_thm}]
    If $\approximationspeed(\fctclass|\approxclass)\leq 0$ then the result is trivial since we always have $\encodingspeed(\fctclass) \geq 0$. In the rest of the proof we assume $\approximationspeed(\fctclass|\approxclass)>0$. Fix $0<\speed'<\min(\approximationspeed(\fctclass|\approxclass),\speed)$ and $h>0$. First,
    $\approxclass$ is $\speed$-encodable so there exists a $(\speed,h)$-encoding of $\approxclass$ that we denote $\approxclass(\speed,h)$. This means that there exist constants $c'_1,c'_2>0$ such that for every $M\in\N$, the set $\approxclass(\speed,h)_M$ is a $c'_1M^{-\speed}$-covering of $\approxclass_M$ of size $|\approxclass(\speed,h)_M|\leq 2^{c'_2M^{1+h}}$. Second, since $0<\speed'<\min(\approximationspeed(\fctclass|\approxclass),\speed)$, the definition of the approximation speed guarantees that there exists a constant $c'_3>0$ such that for every $f\in\fctclass$ and every $M\in\N$, there exists a function $\Phi_M(f)\in\approxclass_M$ that satisfies:
    \[
        d\left(f,\Phi_M(f)\right) \leq c'_3M^{-\speed'}.
    \]
    Since $0< \speed'<\speed$, note that for every $M\in\N$, it holds $c'_1M^{-\speed}+c'_3M^{-\speed'}\leq (c'_1+c'_3)M^{-\speed'}$. Define $c_1=c'_1+c'_3$ and $c_2=c'_2$. We deduce that for every $M\in\N$, the set $\approxclass(\speed,h)_M$ is a $c_1M^{-\speed'}$-covering of $\fctclass$ of size $|\approxclass(\speed,h)_M|\leq 2^{c_2M^{1+h}}$. Now, for every $\eps>0$, the integer $M_\eps:=\left\lceil \left(\frac{c_1}{\eps}\right)^{1/\speed'} \right\rceil$ satisfies $\eps\geq c_1M_\eps^{-\speed'}$. By monotonicity of the metric entropy $H(\fctclass, d, \cdot)$ we get
    $
        H(\fctclass, d, \eps)\leq H(\fctclass, d, c_1M_\eps^{-\speed'}) \leq c_{2}M_{\eps}^{1+h}.
    $
    Note that for $0<\eps<c_1$, denoting by $c= (2c_1^{1/\speed'})^{1+h}$ it holds $M_\eps^{1+h}\leq \left(1+\left(\frac{c_1}{\eps}\right)^{1/\speed'} \right)^{1+h} = \left(\frac{c_1}{\eps}\right)^{(1+h)/\speed'} \left(1+\left(\frac{\eps}{c_1}\right)^{1/\speed'} \right)^{1+h}\leq c\eps^{-(1+h)/\speed'}$. Finally for every $0<\eps<c_1$, it holds
    \[
        H(\fctclass, d, \eps)\leq c\eps^{-(1+h)/\speed'},
    \]
    As a direct consequence of \Cref{def: encoding speed}, this implies $\encodingspeed(\fctclass) \geq \frac{\speed'}{1+h}$ for every $h>0$ and every $0<\speed'<\min(\approximationspeed(\fctclass|\approxclass),\speed)$, hence the desired result.
\end{proof}
\section{First examples of $\infty$-encodable approximation families}\label{app:FirstExInfEnc}
\begin{proof}[Proof of \Cref{lem: some sequences of balls in finite dim are infinite-encodable}]
    Each $\approxclass_{M}$ can be identified with the closed ball of radius $r_M$ in dimension $d_M$ with respect to the $q$-th norm, so that standard bounds on covering numbers \cite[Eq. (5.9)]{Wainwright19HDS} yield for every $0<\eps\leq r_M$:
    \begin{equation}\label{ineq: bounds covering number of balls}
        d_M\log_2\left(\frac{r_M}{\eps}\right)\leq H(\approxclass_M,\|\cdot\|_q,\eps) \leq d_M\log_2\left(\frac{3r_M}{\eps}\right).
    \end{equation}
    For $\eps=M^{-\gamma}(\leq 1\leq r_M)$, we get:
    \[
        d_M(\log_2(r_M)+\gamma\log_2(M))\leq H(\approxclass_M,\|\cdot\|_q,\eps)
        \leq d_M(\log_2(3r_M)+\gamma\log_2(M)).
    \]
    Everything is non-negative, so if the right hand-side is $\mathcal{O}_{M\to\infty}(M^{1+h})$, for every $h>0$, then so is the left hand-side. The converse is also true since both sides only differ by $\log_2(3)d_M=\mathcal{O}_{M\to\infty}(d_M\log M)$. The non-negativity of the quantities also implies that the condition $d_M[\log_2(r_M)+\speed\log_2(M)]=\mathcal{O}_{M\to\infty}(M^{1+h})$, for every $h>0$, does not depend on $\speed$. As a consequence, either $\approxclass$ is $\infty$-encodable or it is never $\gamma$-encodable, whatever $\gamma>0$ is. Finally, note that for every $h>0$, $d_M\left(\log_2(r_M)+\log_2(M)\right)=\mathcal{O}_{M\to\infty}(M^{1+h})$ if and only if $d_M\left(\log_2(r_M)+1\right)=\mathcal{O}_{M\to\infty}(M^{1+h})$. The "only if" part is clear since for $M\geq 2$, it holds $0\leq d_M\left(\log_2(r_M)+1\right)\leq d_M\left(\log_2(r_M)+\log_2(M)\right)$. For the "if" part, use that $r_M\geq 1$ and the assumption to get $0\leq d_M\leq d_M\left(\log_2(r_M)+1\right)=\mathcal{O}_{M\to\infty}(M^{1+h})$ so that $d_M\log_2(M)=\mathcal{O}_{M\to\infty}(M^{1+h}\log_2(M))=\mathcal{O}_{M\to\infty}(M^{1+h})$.
\end{proof}
\begin{proof}[Proof of \Cref{thm:LipsParamAreUnifQuant}]
    Fix an arbitrary $\speed>0$. \Cref{lem:CovNumberLipImage} guarantees that for every integer $M\geq 2$, the set $\varphi_M(\mQ(\approxclass_M,\eta_M(\speed), \infty))$ is an $M^{-\gamma}$-covering of $\varphi_M(\approxclass_M)$ of size satisfying:
    \[
        \log_2(|\varphi_M(\mQ(\approxclass_M,\eta_M(\speed), \infty))|)\leq c(q,\speed)\biggl(d_M\biggl[\log_2(d_M) + \log_2(r_M)+\log_2(\Lips(\varphi_M))+\log_2(M)\biggr]\biggr).
    \]

    Since $0\leq d_M\leq d_M\left(\log_2(r_M)+\log_2(\Lips(\varphi_M))+1\right)$, Assumption \eqref{eq:LipCondition} guarantees that for every $h>0$, it holds $d_M = \mathcal{O}_{M\to\infty}(M^{1+h})$ so that $d_M\left(\log_2(d_M)+\log_2(M)\right)=\mathcal{O}_{M\to\infty}(M^{1+h}\left(\log(M^{1+h})+\log_2(M)\right)) = \mathcal{O}_{M\to\infty}(M^{1+h})$. As a consequence, for every $h>0$, it holds $\log_2(|\varphi_M(\mQ(\approxclass_M,\eta_M(\speed), \infty))|)=\mathcal{O}_{M\to\infty}(M^{1+h})$ so that the sequence $\mQ(\varphi(\approxclass)|\speed)$ is a $(\speed,h)$-encoding of $\varphi(\approxclass)$. This shows that $\varphi(\approxclass)$ is $\speed$-encodable for every $\speed>0$, so it is $\infty$-encodable.
\end{proof}
\section{Encodability of approximation families defined with dictionaries and ReLU networks}\label{app:DictNNInfEnc}
\begin{proof}[Proof of \Cref{prop:KerkyacharianEnc}]
    Fix $M\in\N$ and $f=\sum_{i=1}^Mc_ie_i\in \approxclass_M^q$. Let $0<\lambda<1$. Define $Q_\lambda(f):=\sum_{i=1}^M\textrm{sign}(c_i)\left\lfloor{\frac{c_i}{\lambda}}\right\rfloor\lambda e_i$ with $\textrm{sign}(x)=1$ if $x\geq 0$, $-1$ otherwise. It is proven in {\cite[Prop. 6]{Kerkyacharian04Entropy}} that there exists a constant $c(p,q)>0$ that only depends on $p$ and $q$ such that:
    \[
        \|f-Q_\lambda(f)\|_{\F} \leq c(p,q)\lambda^{1-q/p}\sup_{i\in\N}\|e_i\|_{\F}.
    \]
    Moreover, it is proven in {\cite[Lem. 4 and proof of Prop. 11]{Kerkyacharian04Entropy}} that the family $(Q_\lambda(f))_{f\in\approxclass_M^q}$ has at most $2^{\lambda^{-q}(1-\log_2(\lambda)+\log_2(M))}$ elements. Setting $\eps=\lambda^{1-q/p}$, and observing that $\lambda^{-q}=\eps^{-1/s}$, this proves that the family $(Q_\lambda(f))_{f\in\approxclass_M^q}$ is a $\mathcal{O}_{\eps\to 0}(\eps)$-covering of $\approxclass_M^q$ of size $\mathcal{O}_{\eps\to 0}(2^{\eps^{-1/s}(\log_2 1/\eps +\log_2 M)})$, with constants independent of $M$. For every $M\in\N$, using the above result with $\eps=M^{-s}$ proves that $\approxclass^q$ is $s$-encodable.
\end{proof}
\begin{proof}[Proof of \Cref{thm:GSInfEnc}] Consider $c>0$. We first prove that $\tilde{\approxclass}^{\pi,c}$ is $\infty$-encodable. Consider $M\in\N$, $\mathcal{I}_M:=\{I\subset\{1,\dots,\pi(M)\},$ $ |I|\leq M\}$, and define for each $I\in\mathcal{I}_M$ the set $\tilde{\approxclass}^{\pi,c}(I):=\{\sum_{i\in I}\Tilde{c}_i\Tilde{\phi}^I_i, (\Tilde{c}_i)_{i\in I}\in[-c,c]^I\}$. It holds:
    \[
        \tilde{\approxclass}^{\pi,c}_M=\bigcup_{I\in\mathcal{I}_M}\tilde{\approxclass}^{\pi,c}(I).
    \]
    Since each $I\in\mathcal{I}_M$ is a set of at most $M$ integers between $1$ and $\pi(M)$, one can describe each such set by $M$ sequences of at most $\log_2(\pi(M))$ bits so that there are at most $2^M\pi(M)$ such sets. Moreover, the set $\tilde{\approxclass}^{\pi,c}(I)$ is the image of $\varphi_{M,I}:(\Tilde{c}_i)_{i\in I}\in([-c,c]^I,\|\cdot\|_2)\mapsto \sum_{i\in I}\Tilde{c}_i\Tilde{\phi}^I_i\in\F$. This map is $1$-Lipschitz (since $(\Tilde{\phi}^I_i)_{i\in I}$ is orthonormal). Equation~\eqref{eq: size covering} of \Cref{lem:CovNumberLipImage} with $n=|I| \leq M$, $q=\infty$ and $r=\max(c,1)$ proves that $\tilde{\approxclass}^{\pi,c}(I)$ has an $M^{-\gamma}$-covering with at most $2^{2c(q,\speed)M\log_2(rM)}$ elements. Taking the union of such covers for each $I\in\mathcal{I}_M$, we end up with an $M^{-\gamma}$-covering of the whole set $\tilde{\approxclass}^{\pi,c}_M$ with at most $2^M\pi(M)2^{2c(q,\speed)M\log_2(rM)}=\mathcal{O}_{M\to\infty}(2^{M\log M})$ elements. This proves the $\infty$-encodability of $\tilde{\approxclass}^{\pi,c}$.

    It now remains to prove Equation~\eqref{eq:GSSpeed}. First, for every $c>0$ and every $M\in\N$, it holds $\tilde{\approxclass}^{\pi,c}_M\subset\approxclass^\pi_M$ so that $\approxclass^\pi$ approximates $\fctclass$ at least as quickly as $\tilde{\approxclass}^{\pi,c}$, that is $\approximationspeed(\fctclass|\approxclass^\pi)\geq \approximationspeed(\fctclass|\tilde{\approxclass}^{\pi,c})$. As we now prove, there is actually equality for $c=\sup\limits_{f\in\fctclass}\sup\limits_{M\in\N}\max\limits_{I\in\mathcal{I}_M}\max\limits_{i\in I}|\langle f,\tilde{\phi}^I_i\rangle_\F|$ (and thus for any larger $c$ since $\approximationspeed(\fctclass|\tilde{\approxclass}^{\pi,c})$ is non-decreasing in $c$). Note that by Cauchy-Schwarz, $c\leq \sup_{f\in\fctclass}\|f\|_\F$ which is finite since $\fctclass$ is bounded. If $f\in\fctclass$, then for every $M\in\N$, every $I\subset\{1,\dots,\pi(M)\}, |I|\leq M$, and every $(c_i)_{i\in I}\in\R^I$, it holds:
    \[
        d(f,\tilde{\approxclass}^{\pi,c}_M)\leq \|f-\sum_{i\in I}\langle f,\tilde{\phi}^I_i\rangle_\F \tilde{\phi}^I_i\|_{\F}\leq \|f-\sum_{i\in I}c_i\phi_i\|_{\F}.
    \]
    This implies that $d(f,\tilde{\approxclass}^{\pi,c}_M)\leq d(f,\approxclass^\pi_M)$. As a consequence, $\tilde{\approxclass}^{\pi,c}$ approximates $\fctclass$ at least as quickly as $\approxclass^\pi$, that is $\approximationspeed(\fctclass|\tilde{\approxclass}^{\pi,c})\geq \approximationspeed(\fctclass|\approxclass^\pi)$. This yields equality \eqref{eq:GSSpeed}.
\end{proof}
\begin{proof}[Proof of \Cref{prop:InfEncNN}]
    By definition of $\infty$-encodability, we have to prove that for every $\speed>0$ and for every $h>0$, the quantized sequence $\mQ(\nnclass|\speed)$ is a $(\speed,h)$-encoding of $\nnclass$. Fix $\speed>0$. \Cref{lem:gammaencodNN} proves that there exists a constant $c>0$ such that for every $M\in\N$, the set $\mQ(\nnclass_M|\speed)$ is a $cM^{-\speed}$-covering of $\nnclass_M$, and for each $M\in\N$, each architecture $\archi{L}{N}\in\archiclass{M}$ and each support $S\in S^M_{\archi{L}{N}}$, the quantized set $R_{\mQ(\paramarchiclass{L}{N}^q(r_M), \eta_M, r_M),S}$ has a number of elements that satisfies:
    \[
        \log_2(|R_{\mQ(\paramarchiclass{L}{N}^q(r_M), \eta_M, r_M),S}|) \leq c M\left(\log_2(r_M) + \log_2(\Lips(M,q)) + \log_2 (M)\right).
    \]

    Fix $h>0$. By assumption \eqref{eq:NSCArchGrowth2} of \Cref{prop:InfEncNN} , we deduce that there exists $c'=c'(h)>0$ such that for every $M\in\N$, and every architecture $\archi{L}{N}\in\archiclass{M}$, we have
    \[
        M\left(\log_2(r_M) + \log_2(\Lips(M,q)) + \log_2(M)\right)
        \leq c' M^{1+h}.
    \]
    Thus, the quantized set
    $\mQ(\nnclass_M|\speed)$  is a $cM^{-\gamma}$-covering of $\nnclass_M$ and its cardinality satisfies
    \[
        |\mQ(\nnclass_M|\speed)| \leq
        \sum\limits_{\archi{L}{N}\in\archiclass{M}}\sum\limits_{S\in S^M_{\archi{L}{N}}} |R_{\mQ(\paramarchiclass{L}{N}^q(r_M), \eta_M, r_M), S}|
        \leq |\archiclass{M}| \cdot |S^M_{\archi{L}{N}}| \cdot 2^{{cc'M^{1+h}}}.
    \]
    Note that for every $M\in\N$, $|\archiclass{M}|\leq L_M M^{L_M-1}$ (at most $L_M$ possibilities for the depth and then, $M$ possibilities for each of the potential $L_M-1$ intermediary layers, the size of the input and output being fixed  to $d_{\textrm{in}}$ and $d_{\textrm{out}}$).
    Similarly, since $S^M_{\archi{L}{N}}$ consists at most of all the supports of size $M$ in dimension
    $d_{\archi{L}{N}}\leq 2M^2L_M$, its cardinality is bounded by $(2M^2L_{M})^M$. Overall, we obtain that
    \[
        \log_{2}(|\mQ(\nnclass_M|\speed)|)
        \leq
        \log_{2}(L_{M})+L_{M}\log_{2}(M) + M \log_{2}(2M^2L_{M}) + cc' M^{1+h}.
    \]
    Using assumption \eqref{eq:NSCArchGrowth2} again, we obtain that there exists $c''>0$ such that
    $\log_{2}(|\mQ(\nnclass_M|\speed)|) \leq c''M^{1+h}$ for every $M \in \N^{*}$. We deduce that for every $\speed>0$ and for every $h>0$, the sequence $\mQ(\nnclass| \speed)$ is a $(\speed,h)$-encoding of $\nnclass$. This yields the result.
\end{proof}
\end{document}